\documentclass[prx,aps,twocolumn,showpacs,floatfix]{revtex4-2}
\usepackage{amsmath}
\usepackage{amssymb}
\usepackage{cmap}
\usepackage[T1]{fontenc}
\usepackage[utf8]{inputenc}
\usepackage{lmodern}
\usepackage{amsfonts}
\usepackage{amsthm}
\usepackage{amsopn}
\usepackage{bm}
\usepackage{booktabs}
\usepackage{dcolumn}
\usepackage{graphicx}
\usepackage{mathtools}
\usepackage{tikz-cd}
\usepackage{xcolor}
\usepackage{hyperref}

\graphicspath{{../Figures/}{../../Figures/}}
\hypersetup{
colorlinks=true,
linkcolor=blue,
citecolor=blue,
urlcolor=blue,
linktoc=page
}
\newtheorem{theorem}{Theorem}[section]
\newtheorem{lemma}[theorem]{Lemma}
\newtheorem{proposition}[theorem]{Proposition}

\newcommand{\F}{\mathbb F_2}
\newcommand{\imop}{\operatorname{im}}
\newcommand{\kerop}{\operatorname{ker}}
\newcommand{\coker}{\operatorname{coker}}
\newcommand{\rank}{\operatorname{rank}}

\newcommand{\grade}{\operatorname{grade}}

\begin{document}

\title{Topological Codes from Space Groups: A Route beyond Translation Invariance}
\author{Chong-Yuan Xu$^{1}$}
\author{Ze-Chuan Liu$^{1}$}
\author{Yong Xu$^{1,2}$}
\email{yongxuphy@tsinghua.edu.cn}
\affiliation{$^{1}$Center for Quantum Information, IIIS, Tsinghua University, Beijing 100084, People's Republic of China}
\affiliation{$^{2}$Hefei National Laboratory, Hefei 230088, People's Republic of China}

\begin{abstract}
Translation invariance underlies all algebraic constructions of topological codes with geometrical locality.
It has remained an open question whether codes that generically break this invariance 
can still be topological and simultaneously possess geometrical locality. 
Resolving this question is important both fundamentally---deepening 
our understanding of topological phases---and practically, as relaxing 
translation invariance could vastly expand the design space and potentially reduce resource 
overhead in fault-tolerant architectures. 
Here we introduce space-group codes, in which crystallographic point-group operations 
enter the bulk stabilizer algebra; bivariate bicycle (BB) codes arise as the translation-only limit. 
The key insight is that the point-group orbit resolves topology and locality together: 
it yields a computable algebraic criterion for topological order and a folded geometry in which point-group operations become local. 
We identify space-group codes whose code parameters exceed the reported same-blocklength, same-check-weight BB benchmarks. 
In five parameter-matched neutral-atom comparisons, reflection codes reduce the optimized movement cost in every
case, by up to $60\%$, while folded placements also enable lower-overhead multilayer superconducting layouts. 
Treating spatial operations as a code-design variable therefore opens a route to topological codes jointly optimized for information protection and hardware geometry.
\end{abstract}
\maketitle

\setcounter{tocdepth}{2}
\tableofcontents

\section{Introduction}
\label{sec:introduction}

\begin{figure*}[htbp]
\includegraphics[width=\linewidth]{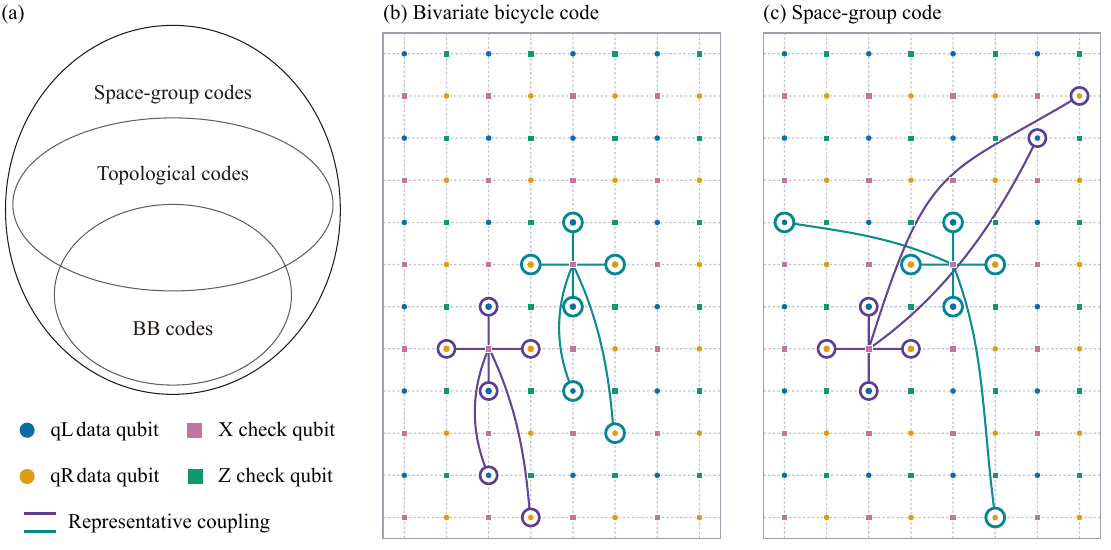}
\caption{
(a) Schematic organization of the code families considered in this work. 
Space-group codes form the broader design space, while BB codes constitute the translation-generated sector. 
Their overlap with the topological-code region represents topological BB codes, whereas the portion beyond 
the BB sector represents topological space-group codes enabled by non-translation space-group operations. 
(b),(c) Blue and orange circles denote the two data-qubit species, while magenta and green squares denote $X$ and $Z$ type checks, respectively. 
Light-gray edges show the nearest-neighbor couplings, and purple and teal edges highlight the neighborhoods of two representative checks. 
In the BB code, the purple check can be obtained by translating the green check; in the space-group code, however, this is not the case.
The BB code and the space-group code are defined by $(f,h)=(1+t_y^{-1}+t_y^{-2},1+t_x^{-1}+t_y^{-2})$ and
$(f,h)=(1+t_y^{-1}+s_yt_x^{-2},1+t_x^{-1}+s_xt_y^{-3})$ [see Eq.~(\ref{eq:stabilizergenerators})], respectively, 
where $t_x,t_y$ denote lattice translations and $s_x,s_y$ denote reflections.
}
\label{fig:framework}
\end{figure*}

The idea of topological quantum codes is rooted in topological 
order~\cite{wen1990topological,wen1993topological,levin2006detecting,wen2013topological} and was systematically 
formulated in Kitaev's seminal work~\cite{kitaev2002topological,kitaev2003fault}, 
first circulated in 1997 and later published as \textit{Fault-Tolerant Quantum 
Computation by Anyons}~\cite{kitaev2003fault}. 
By combining locally measurable stabilizer constraints with the nonlocal storage 
of quantum information, topological codes provide a natural route toward scalable 
fault-tolerant quantum computation~\cite{dennis2002topological,bombin2013introduction,fujii2015quantum}.
Specifically, in a scalable family of topological codes, logical operators span 
increasingly large spatial regions as system sizes grow, and the code distance increases accordingly.
When the physical error rate lies below the fault-tolerance threshold, 
the logical error rate is exponentially suppressed with the code distance.
Increasing the system size therefore provides a natural and direct scaling 
route toward fault-tolerant quantum computation~\cite{bombin2013introduction}.

Early studies of topological codes focus primarily on toric codes~\cite{kitaev2003fault,kitaev1997quantumcomputations,kitaev1997quantumerrorcorrection,joyner2004toric,castelnovo2013negativity,araujo2020pedagogical}, surface codes~\cite{bravyi1998quantum,fowler2012surface,horsman2012surface}, and color codes~\cite{steane1996error,bombin2006topological,landahl2011fault,kesselring2024anyon}.
More recent developments have introduced broader families, including bivariate bicycle (BB) codes~\cite{bravyi2024high,liang2025generalized,chen2025anyon,kobayashi2026generalized}, tile codes~\cite{liang2025planar,steffan2025tile,breuckmann2025logical}, 
and multivariate bicycle codes~\cite{voss2025multivariate,lin2026abelian,mian2026multivariate}, yielding practical finite-size instances with improved encoding rates and code distances.
A common organizing principle of these periodic algebraic constructions is 
that their bulk stabilizer structures are generated by lattice translations;
this naturally enforces geometrical locality and facilitates hardware implementation.
It raises a fundamental question: is translation invariance essential 
to the systematic construction of topological codes that preserve geometrical locality? 
At first glance, significantly relaxing translation invariance appears to threaten geometrical locality. 
Indeed, certain good quantum low-density parity-check (qLDPC) codes, 
such as quantum Tanner codes~\cite{leverrier2022quantum}, can exhibit the strong form of topological order~\cite{anshu2023nlts}
in terms of  No Low-Energy Trivial State~\cite{Freedman2014}, yet they lack geometrical locality.
The challenge, therefore, is to determine whether a systematic framework exists 
for constructing quantum codes beyond translation invariance that retain both 
topological order and geometrical locality.

In this work, we answer this question in the affirmative.
This work, to the best of our knowledge, provides the first systematic construction and certification, through a computable criterion for topological order, of a class of non-translation-invariant topological codes whose construction incorporates nontrivial crystallographic point-group operations while retaining geometric locality by orbifold.
We first construct a general class of Calderbank--Shor--Steane (CSS) stabilizer codes from permutation-group actions and then specialize the construction to symmorphic space groups.
Such a space group can be written as $G=T\rtimes P$, where $T$ is a lattice-translation group and $P$ 
is a crystallographic point group generated by finite-order rotations and reflections.
When $P$ is nontrivial, rotations and reflections generate space-group codes beyond the purely translation-invariant setting, as illustrated in Fig.~\ref{fig:framework}(c).
When $P$ is trivial, the construction reduces to the BB codes, as illustrated in Fig.~\ref{fig:framework}(b).
BB codes therefore arise as a special case of space-group codes within the present framework, as illustrated in Fig.~\ref{fig:framework}(a).
To characterize this enlarged family, we develop an algebraic formulation based on invariance theory.
On this basis, we establish a computable criterion for determining whether a 
space-group code satisfies the topological-order condition and extend Gr\"obner-basis methods to the computation of anyon-counting for space-group codes.

We further introduce an orbifold description that connects point-group actions to the physical spatial geometry of the codes.
At first glance, stabilizer terms involving reflections or rotations may appear to require long-range connections on the order of the system size. 
In orbifold coordinates defined by point-group identifications, however, the same group actions can acquire geometrically local representations, which lead to more local geometric operations.
In particular, reflection codes can be reorganized into open-boundary layouts in the same spatial dimension, such that couplings that appear long ranged in the periodic representation become local in the folded geometry.
Numerical searches further demonstrate the practical consequences of this enlargement.
We search and identify a collection of finite-size CSS codes with competitive code parameters beyond BB codes with the same number of physical qubits and stabilizer weight.
We also investigate the adaptation of space-group codes to two quantum-computing platforms: reconfigurable neutral-atom arrays and superconducting circuits.
Space-group codes and their orbifold descriptions effectively reduce the travel distance required for acousto-optic-deflector (AOD) transport, thereby helping preserve the available coherence-time budget.
For routing on superconducting circuits, the corresponding orbifold layouts substantially reduce the mean coupler length in the instances studied, which can ease physical implementation and improve
coupling fidelity.
Our results provide a concrete blueprint for future experiments and open a new avenue for the co-design of quantum codes and quantum-computing platforms.

The remainder of this paper is organized as follows.
Section~\ref{sec:background} reviews CSS codes and the required
preliminaries from group algebra, module theory, and commutative algebra.
Section~\ref{sec:construction} presents the general construction of CSS
codes from permutation-group actions.
Section~\ref{sec:topological-sg} specializes the construction to space
groups and develops the corresponding theory of geometric locality,
invariant modules, and topological order.
Section~\ref{sec:movement-illustration} translates the folded locality of
reflection codes into movement protocols for reconfigurable atom arrays.
Section~\ref{sec:numerical-results} presents numerical results on code
scaling, finite-size parameters, AOD transport,
and superconducting layouts.
Finally, Sec.~\ref{sec:discussion} summarizes the main implications and
scope of our results and discusses directions for future work.

\section{Background and notation}
\label{sec:background}

The central mathematical framework underlying this work is 
a correspondence between CSS codes and algebraic objects 
such as groups, rings, and modules.
In this section, we provide a brief review of the CSS stabilizer 
formalism and the requisite algebraic preliminaries.
For a more detailed introduction, we refer the readers to Refs.~\cite{nielsen2010quantum,atiyah2018introduction}. 

\subsection{CSS codes}
\label{subsec:css-codes}

CSS codes are a special class of stabilizer codes with a stabilizer group $\mathcal S$ 
generated by a set $S_X$ of $X$-type stabilizer generators and a set $S_Z$ of $Z$-type
stabilizer generators~\cite{steane1999enlargement,gottesman1997stabilizer,nielsen2010quantum}.
Such a code is determined by the chain complex~\cite{freedman2001projective,bombin2007homological,breuckmann2021quantum}
\begin{equation}
	\begin{tikzcd}[column sep=large,row sep=large]
		\F^{n_X}
		\arrow[r,"H_X^T"]
		&
		\F^{n}
		\arrow[r,"H_Z"]
		&
		 \F^{n_Z}
	\end{tikzcd},
	\label{eq:complex}
\end{equation}
together with its dual in Eq.~(\ref{eq:complexdual}).
Here, $H_X$ and $H_Z$ are $X$- and $Z$-type check matrices, respectively, 
$n_X$ ($n_Z$) denotes the number of $X$-type ($Z$-type) stabilizer generators, 
and $n$ is the number of physical qubits.
Note that the rows in $H_X$ and $H_Z$ are not necessarily linearly independent.
In addition, the requirement that every $X$-type generator commutes with every $Z$-type 
generator is equivalent to the chain-complex condition $H_ZH_X^T=0$.

A binary vector in $\F^{n}$ represents an $X$-type Pauli string 
operator, while a vector in $\F^{n_X}$ ($\F^{n_Z}$) corresponds to a product of 
$X$-type ($Z$-type) stabilizer generators. Let $H_X^T=\left(\begin{array}{cccc}
	u_1 & u_2 & \dots & u_{n_X}
\end{array}\right)$, where each column $u_i$ in $\F^{n}$ is the binary 
representation of the $i$th $X$-type generator. Then a vector $v_X$ in $\F^{n_X}$  
is mapped to an $X$-type Pauli operator $\sum_{i=1}^{n_X}[v_X]_i u_i$, so
the image $\imop(H_X^T)$ is exactly the group of $X$-type stabilizers generated by $S_X$. 
The syndrome of an $X$-type Pauli operator $v$ in $\F^{n}$ is obtained by 
applying the check matrix $H_Z$, yielding $H_Z v \in \F^{n_Z}$; the support of 
this syndrome indicates precisely which $Z$-type stabilizer generators are violated.
Consequently, the kernel $\kerop(H_Z)$ consists of all $X$-type Pauli operators 
that commute with every $Z$-type stabilizer. This subspace contains
$X$-type stabilizers and $X$-type logical operators.

The dual of the chain complex in Eq.~(\ref{eq:complex}) is given by
\begin{equation}
	\begin{tikzcd}[column sep=large,row sep=large]
		\F^{n_X}
		&
		\arrow[l,"{H_X}"']
		\F^{n}
		&
		\arrow[l,"{H_Z^T}"']
		\F^{n_Z}
	\end{tikzcd}.
	\label{eq:complexdual}
\end{equation}
The preceding analysis carries over directly upon interchanging the roles of 
$X$ and $Z$.

For an infinite system, a CSS code is topological in the algebraic sense 
if its chain complex and dual satisfy the middle exactness conditions~\cite{haah2013commuting}:
\begin{equation} \label{eq:topological-exactness-background}
	\begin{split}
		\imop(H_X^T) &= \kerop(H_Z),\\
		\imop(H_Z^T) &=\kerop(H_X).
	\end{split}
\end{equation}
They indicate that every finite-support Pauli operator 
that commutes with all check operators is itself generated by finite-support stabilizers.
In other words, the stabilizer generators must be local in the relevant physical or 
geometric metric, whereas all logical operators are nonlocal.
Recently, the conditions are used as the starting point for classification of anyons 
in generalized toric codes and BB codes~\cite{liang2025generalized}.

\subsection{Algebra}
\label{subsec:algebra}

In this subsection, we briefly review the algebraic language used 
repeatedly in this paper: free modules and group 
algebras~\cite{gallian2021contemporary,eisenbud1995commutative}.
Unless stated otherwise, all vector spaces are over $\F$.

Given a ring $\mathcal R$ and a positive integer $n$, we use $\mathcal R^n$ 
to denote a free $\mathcal R$-module of rank $n$~\cite{eisenbud1995commutative}:
\begin{equation}
	\mathcal R^n=\{(r_1,\ldots,r_n)^T\mid r_i\in\mathcal R\}.
	\label{eq:free-module}
\end{equation}
Scalar multiplication and addition for $R$-vectors in the module are defined as 
follows:
\begin{equation}
	\begin{split}
    &r(\dots,r_i,\dots)^T=(\dots,r r_i,\dots)^T,\\
    &(\dots,r_i,\dots)^T+(\dots,r'_i,\dots)^T=(\dots,r_i+r'_i,\dots)^T,
	\end{split}
\end{equation}
which are the same as the operations on vectors in a vector space with entries 
in a field.

For a group $G$, the group algebra $\F[G]$ is defined as~\cite{gallian2021contemporary}:
\begin{equation}
	\F[G]=\bigg\{\sum_{g\in G}\mu_g g\ \bigg|\ \mu_g\in\F\bigg\}.
	\label{eq:group-algebra}
\end{equation}
It is a vector space with a basis $G$ if only scalar multiplication and 
addition are considered. For the group algebra, multiplication for any two elements
\begin{equation}
	f_1=\sum_{g\in G}\mu_g g,\qquad
	f_2=\sum_{g'\in G}\lambda_{g'} g',
\end{equation}
in $\F[G]$ is defined as
\begin{equation}
	f_1f_2=\sum_{g,g'\in G}(\mu_g\lambda_{g'})(gg').
	\label{eq:group-algebra-product}
\end{equation}
The formal transpose of any element $f_1$ is defined as
\begin{equation}
	f_1^T=\sum_{g\in G}\mu_g g^{-1},
	\label{eq:formal-transpose}
\end{equation}
which constitutes the inversion anti-involution of the group algebra:
\begin{equation}
	(f_1f_2)^T=f_2^Tf_1^T,\qquad (f_1^T)^T=f_1.
	\label{eq:formal-transpose-anti-involution}
\end{equation}

\section{Code construction from symmetric groups}
\label{sec:construction}

In this section, we present how to construct CSS codes based on
symmetric groups.

\subsection{A guiding example}
\label{subsec:construction-guiding-example}

Let us first consider a simple example for illustration. Consider a symmetric group 
$S_4$ over a set  $A_{4}=\{a_1,a_2,a_3,a_4\}$
with four elements. Acting a permutation operation $g$ in $S_4$ on an element $a_i$ leads
to $g a_i=a_j$ so that the matrix representation $\Gamma(g)$ of $g$ is given by
\begin{equation}
	\Gamma(g)_{i' i}=\delta_{i' j}.
	\label{eq:permutation-matrix-entry}
\end{equation}
Choose two permutation elements from $S_4$, e.g., 
\begin{equation}
	g_1=(3,4,1,2),\qquad g_2=(2,1,4,3),
\end{equation}
whose matrix representations are given by
\begin{equation}
	\Gamma(g_1)=
	\begin{pmatrix}
		0&0&1&0\\
		0&0&0&1\\
		1&0&0&0\\
		0&1&0&0
	\end{pmatrix},
	\qquad
	\Gamma(g_2)=
	\begin{pmatrix}
		0&1&0&0\\
		1&0&0&0\\
		0&0&0&1\\
		0&0&1&0
	\end{pmatrix},
\end{equation}
respectively.
Setting $C=\Gamma(g_1)+\Gamma(g_2)$,
we define the check matrices as 
\begin{equation}
	H_X=H_Z=\left(C\mid C\right)=
	\begin{pmatrix}
		0&1&1&0&0&1&1&0\\
		1&0&0&1&1&0&0&1\\
		1&0&0&1&1&0&0&1\\
		0&1&1&0&0&1&1&0
	\end{pmatrix}.
	\label{eq:toy-check-matrix}
\end{equation}
We clearly see that
$
	H_XH_Z^T=0,
$
and equivalently, $H_ZH_X^T=0$, so it defines a CSS code.

\subsection{General case}
\label{subsec:operator-group}

Let $S_n$ be a symmetric group on a set $A_n=\{a_1,a_2,\dots,a_n\}$ 
(see Appendix~\ref{app:generate-set} using coset space of a group for code construction).
To construct a CSS code, we choose four group-algebra 
elements $f_X,h_X,f_Z,h_Z\in\F[S_n]$ and define the check matrices as
\begin{equation} \label{eq:gamma-block-maps}
	\begin{split}
	H_X&=\bigl(\Gamma(f_X)\mid \Gamma(h_X)\bigr),\\
	H_Z&=\bigl(\Gamma(f_Z)\mid \Gamma(h_Z)\bigr).
	\end{split}
\end{equation}

\begin{lemma}
	\label{thm:block-transpose}
	The transpose of the check matrices in Eqs.~(\ref{eq:gamma-block-maps}) are given by
\begin{equation}
	\begin{split}
	H_X^T&=\bigl(\Gamma(f_X^T)\mid \Gamma(h_X^T)\bigr)^T,\\
	H_Z^T&=\bigl(\Gamma(f_Z^T)\mid \Gamma(h_Z^T)\bigr)^T.
	\end{split}
\end{equation}
\end{lemma}

\begin{proof}
	For any group element $g\in S_n$, $\Gamma(g)$ is a permutation matrix, so
	\begin{equation}
		\Gamma(g^{-1})=\Gamma(g)^T.
	\end{equation}
	We thus have
	\begin{equation}
		\Gamma(f^T)=\Gamma(f)^T,\qquad \Gamma(h^T)=\Gamma(h)^T.
		\label{eq:gamma-formal-ordinary-transpose}
	\end{equation}
	Therefore, Eqs.~(\ref{eq:gamma-block-maps}) hold.
\end{proof}

The CSS commutation condition imposes the constraints on the group-algebra 
elements:

\begin{theorem}
	\label{thm:faithful-equivalence}
	Assume that the permutation matrices of $S_n$ on $A_n$ form a 
	faithful representation of the group algebra $\F[S_n]$.
	Then
	\begin{equation}
		f_Xf_Z^T+h_Xh_Z^T=0
	\end{equation}
	if and only if
$
		H_XH_Z^T=0.
$

\end{theorem}

\begin{proof}
	For any group elements $g,k\in S_n$, we have 
	\begin{equation}
		\Gamma(k)^T=\Gamma(k^{-1}),\qquad
		\Gamma(g)\Gamma(k)=\Gamma(gk),
	\end{equation}
	and thus
	\begin{equation}
		\Gamma(g)\Gamma(k)^T=\Gamma(gk^{-1}).
	\end{equation}
	Based on these equations, we obtain that
	\begin{align}
		H_XH_Z^T
		&=\Gamma(f_X)\Gamma(f_Z)^T+\Gamma(h_X)\Gamma(h_Z)^T\notag\\
		&=\Gamma(f_Xf_Z^T+h_Xh_Z^T).
	\end{align}
	
	If $f_Xf_Z^T+h_Xh_Z^T=0$, then $H_XH_Z^T=0$.
	Conversely, faithfulness means that $\Gamma(f)=0$ implies $f=0$, so if $H_XH_Z^T=0$, faithfulness of $\Gamma$ gives $f_Xf_Z^T+h_Xh_Z^T=0$.
\end{proof} 

As an illustrative example, consider the toric code with $f_X=1+t_x$, $h_X=1+t_y$, $f_Z=1+t_y^{-1}$, 
and $h_Z=1+t_x^{-1}$, where $t_x$ and $t_y$ are translation operators along the $x$ 
and $y$ directions, respectively. With these definitions, we have 
$f_Xf_Z^T+h_Xh_Z^T=0$ so that the code is a valid CSS code.

\section{Topological space-group codes}
\label{sec:topological-sg}

We now consider $A$ as an infinite lattice and use space group 
$G=T \rtimes P$ that combines lattice translations with 
finite point-group operations to construct topological codes.
The goal of this section is to analyze the locality and topology of 
the space-group codes using ring-module algebra.

\subsection{Setups}
\label{subsec:setups}

We first construct the infinite space on which the space group acts.
Let
\begin{equation}
	\mathcal R_0=\F[x_1^{\pm1},\ldots,x_D^{\pm1}]
\end{equation}
be the Laurent polynomial ring in $D$ variables.
Let $A$ be the set of all Laurent monomials in $\mathcal R_0$.
The identification
\begin{equation}
	x^z=x_1^{z_1}\cdots x_D^{z_D},\qquad z=(z_1,\ldots,z_D)\in\mathbb Z^D
	\label{eq:laurent-monomial}
\end{equation}
identifies $A$ with the infinite integer lattice $\mathbb Z^D$.
An element $q\in\mathcal R_0$ has the finite form
$
	q=\sum_{z\in\mathbb Z^D}c_zx^z
$
with $c_z\in\F$ and only finitely many nonzero coefficients.
An element $p\in\mathcal R_0^2$ is written as a block vector
\begin{equation}
	p=
	\begin{pmatrix}
		q_1\\
		q_2
	\end{pmatrix},
	\qquad q_1,q_2\in\mathcal R_0.
\end{equation}

The translation group 
$
	T\cong\mathbb Z^D.
$
For $b\in\mathbb Z^D$, the corresponding translation acts on $q=\sum_zc_zx^z$ by
$
	t_b(q)=\sum_zc_zx^{z+b}.
$
An element $p_\rho \in P$
acts on $q=\sum_zc_zx^z$ by
\begin{equation}
	p_\rho(q)=\sum_zc_zx^{\rho z},
	\label{eq:point-group-action-r0}
\end{equation}
where $\rho$ is its matrix representation.
In the Euclidean interpretation of a crystallographic point group, the matrices $\rho$ also preserve a lattice metric, but the algebraic construction below only uses the integral finite-order action on $\mathbb Z^D$.

For the space group $G$, we write each element as 
$g=t_b p_\rho$ with $b\in\mathbb Z^D$ and $\rho\in P$. Acting it on $q$ leads to
\begin{equation}
	g(q)=\sum_z c_z x^{\rho z+b}.
	\label{eq:affine-action}
\end{equation}
We define the action of a group-algebra element
$
	f=\sum_{g\in G}\mu_g g\in\F[G]
$
on $q$ by linear extension as 
\begin{equation}
	\Gamma (f)(q)=\sum_{g\in G}\mu_g \,g(q).
	\label{eq:group-algebra-action-r0}
\end{equation}

\begin{theorem}[Faithfulness of the space-group algebra action]
	\label{thm:space-group-faithful}
	The affine lattice action of $G=T\rtimes P$ on the Laurent monomial 
	basis induces an injective group-algebra representation.
	Thus the faithful representation used below is the faithful 
	action of the full space-group algebra.
\end{theorem}

\begin{proof}
	For any $f_1,f_2\in\F[G]$, the matrix representation of $f_1+f_2$, $f_1$, and $f_2$
	satisfy that
	\begin{equation}
		\Gamma(f_1+f_2)=\Gamma(f_1)+\Gamma(f_2).
		\label{eq:gamma-module-homomorphism}
	\end{equation}
	Faithfulness means that $\Gamma$ is injective, which is equivalent to the kernel test
	\begin{equation}
		f\ne0\quad\Longrightarrow\quad \Gamma(f)\ne 0.
		\label{eq:faithful-kernel-test}
	\end{equation}
	Let
	$
		f=\sum_{i} \mu_i g_i\in \F[G]
	$
	be a nonzero group-algebra element, written with distinct group elements and nonzero coefficients.
	Writing $g_i=t_{b_i}p_{\rho_i}$, we have
	\begin{equation}
		\Gamma(f)x^z=\sum_{i} \mu_i x^{\rho_i z+b_i}.
		\label{eq:gamma-f-on-monomial}
	\end{equation}
	If the exponent vectors $\rho_i z+b_i$ are pairwise distinct, then the right-hand side is a 
	nonzero linear combination of distinct Laurent monomial basis vectors and therefore is nonzero.
	Thus it is enough to find one lattice point $z_0$ at which no two terms coincide.
	Equivalently, if $\Gamma(f)x^z=0$, then for the $i$th term, it is either zero or there exists 
	at least one other term $j\ne i$ with
	\begin{equation}
		\rho_i z+b_i=\rho_j z+b_j.
	\end{equation}
	This is because the coefficient of each Laurent monomial basis vector must vanish in $\F$.
	
	For every pair $i\ne j$, we define
	\begin{equation}
		\omega_{ij}=\rho_i-\rho_j,\qquad c_{ij}=b_j-b_i .
	\end{equation}
The $i$th and $j$th terms coincide at $z$ precisely when
\begin{equation}
	\omega_{ij}z=c_{ij}.
	\label{eq:collision-linear-equation}
\end{equation}
	Let
	\begin{equation}
		C_{ij}=\{z\in\mathbb Z^D\mid \omega_{ij}z=c_{ij}\}
	\end{equation}
	be the corresponding coincidence set.
	We need to show that the finite union $\bigcup_{i<j}C_{ij}$ does not fill the whole lattice.

	If $\omega_{ij}=0$, then $\rho_i=\rho_j$.
	Because the group elements $(b_i,\rho_i)$ and $(b_j,\rho_j)$ are distinct, we have $b_i\ne b_j$, hence $c_{ij}\ne0$.
	In this case, $C_{ij}$ is empty.
	It remains to treat the pairs with $\omega_{ij}\ne0$.
	Choose a vector $v$ such that
	\begin{equation}
		\omega_{ij}v\ne0
	\end{equation}
	for every nonzero matrix $\omega_{ij}$ that occurs.
	Such a $v$ exists.
	Indeed, for each nonzero $\omega_{ij}$, choose one nonzero row $w_{ij}$.
	Take
	\begin{equation}
		v=(1,N,N^2,\ldots,N^{D-1}).
	\end{equation}
	Then $w_{ij}\cdot v$ is a nonzero polynomial in $N$, so it has only finitely many integer roots.
	Since there are only finitely many pairs, an integer $N$ can be chosen so that $w_{ij}\cdot v\ne0$ for all selected rows.
	Then $\omega_{ij}v\ne0$ for all nonzero $\omega_{ij}$.
	
	Now restrict to an integer line
	\begin{equation}
		L=\{tv\mid t\in\mathbb Z\}.
	\end{equation}
	For a fixed pair with $\omega_{ij}\ne0$, the coincidence equation on this line is
	\begin{equation}
		t\omega_{ij}v=c_{ij}.
	\end{equation}
	This equation has at most one integer solution $t$: if both $t$ and $t'$ solved it, then
	\begin{equation}
		(t-t')\omega_{ij}v=0,
	\end{equation}
	and $\omega_{ij}v\ne0$ forces $t=t'$.
	Thus each coincidence set $C_{ij}$ intersects $L$ in at most one point.
	There are only finitely many pairs, so the finite union $\bigcup_{i<j}C_{ij}$ intersects $L$ in only finitely many points.
	Since $L$ contains infinitely many lattice points, we can choose $z_0=t_0v\in L$ outside this union.
	For this $z_0$, Eq.~\eqref{eq:collision-linear-equation} cannot be satisfied, 
	so the exponent vectors $\rho_i z_0+b_i$ are pairwise distinct.
	
	Applying Eq.~\eqref{eq:gamma-f-on-monomial} at $z_0$ gives
	\begin{equation}
		\Gamma(f)x^{z_0}=\sum_{i} \mu_i x^{\rho_i z_0+b_i}\ne0,
	\end{equation}
	because the monomials appearing on the right are distinct basis vectors of $\mathcal R_0$ and every $\mu_i$ is nonzero.
	Thus $\Gamma(f)\ne0$ for every nonzero $f\in\F[G]$.
	By Eq.~\eqref{eq:faithful-kernel-test}, $\Gamma$ is faithful as a representation of the full space-group algebra $\F[G]$.
\end{proof}

\subsection{Locality}
\label{subsec:sg-locality-main}

A topological code requires its stabilizer generators to be local in some physical geometry.
For ordinary BB codes, this geometry is usually the periodic translation lattice itself.
For space-group codes, however, stabilizer terms may also contain point-group operations, and we therefore need a natural local geometry in which these point-group operations are also local.
Our approach is to introduce a folded lattice: The original lattice cells are grouped into folded supercells according to the orbits of the point group $P$, and locality is measured by the check-data coupling distance in this folded geometry.
As we explain below, the folded lattice can be realized inside a Euclidean space.

We first use reflection codes as the main example.
In one dimension, the reflection-code space group is generated by a translation $t_x$ and a reflection $s_x$.
In two dimensions, the reflection-code space group is generated by translations $t_x,t_y$ and reflections $s_x,s_y$, where
\begin{equation}
	s_x:x^ay^b\mapsto x^{-a}y^b,\quad s_y:x^ay^b\mapsto x^ay^{-b}.
\end{equation}
We now show that, for one- and two-dimensional reflection codes, both reflection and translation operations can be embedded as local geometric operations in a Euclidean space.

\begin{figure*}[htbp]
	\centering
	\includegraphics[width=\textwidth]{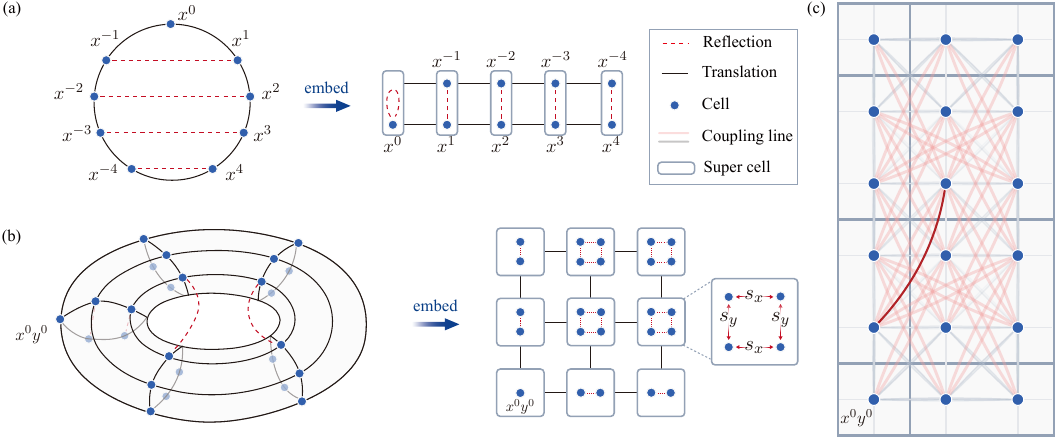}
	\caption{
		Folded locality for reflection and space-group codes.
		(a) A one-dimensional reflection example. Periodic cells are grouped into reflection orbits
		$A_0=\{x^0\}$ and $A_a=\{x^a,x^{-a}\}$, producing a one-dimensional open-boundary folded lattice.
		(b) A two-dimensional reflection example.
		A periodic square lattice is folded by $s_x$ and $s_y$ into supercells given by reflection orbits.
		Interior supercells contain four original cells, mirror-axis supercells contain two cells, and the origin contains one cell.
		(c) A practical folded-layout locality example for the reflection code
		$f=t_xs_y+s_xs_yt_y+s_xt_xs_yt_y$ and
		$h=t_x+s_xt_y^5+s_xt_xt_y$.
		Each cell contains two data qubits and two check qubits, placed on two layers for syndrome extraction.
		The highlighted edge is a longest coupling edge in the folded layout, and its length defines the folded locality.
	}
	\label{fig:embedding}
\end{figure*}

Figure~\ref{fig:embedding}(a) shows the folding of a one-dimensional reflection code.
The original periodic cells are labeled by $x^{-4},x^{-3},\ldots,x^{-1},x^0,x^1,\ldots,x^4$.
The reflection $x^m\mapsto x^{-m}$ partitions these cells into orbits
\begin{equation}
	A_0=\{x^0\},\quad A_a=\{x^a,x^{-a}\},\quad a=1,\ldots,4.
\end{equation}
Each orbit defines a folded supercell.
The fixed point $x^0$ forms a one-cell supercell, while all other cells are paired into two-cell supercells.
Thus the original periodic chain is embedded as a one-dimensional open-boundary folded lattice.
In this folded lattice, reflection exchanges internal labels within a supercell, while translation acts either within a supercell or between neighboring supercells.

Figure~\ref{fig:embedding}(b) shows the folding of a two-dimensional reflection code.
For a $5\times5$ periodic square, the original cells are labeled by $x^ay^b$.
The two reflections $s_x,s_y$ group the cells into orbits
\begin{equation}
	C_{a,b}=\{x^{\epsilon a}y^{\eta b}\mid \epsilon,\eta=\pm 1\}.
\end{equation}
If $a=0$ or $b=0$, the supercell lies on a mirror axis and contains only two cells.
The origin $x^0y^0$ is fixed by both reflections, and hence contains only one cell.
The two-dimensional periodic lattice is represented as a two-dimensional open-boundary folded lattice.
This folding places reflection-related cells in the same supercell.

The same construction applies to general space-group codes.
Let $G=T\rtimes P$.
We group the original lattice cells by $P$-orbits:
\begin{equation}
	\lambda\sim p\lambda,\quad p\in P.
\end{equation}
Each orbit $[\lambda]$ defines a folded supercell.
Point-group operations become internal label permutations within folded supercells, while translations become moves between folded supercells.

For a $D$-dimensional reflection code, this folded lattice embeds directly into a $D$-dimensional open-boundary Euclidean space.
At the continuous level, coordinate reflections give the quotient
\begin{equation}
	\mathbb R^D/(\mathbb Z_2)^D\cong [0,\infty)^D.
\end{equation}
Equivalently, each coordinate direction is folded from a two-sided direction into a half-axis.
For finite lattices, we use the corresponding discrete folded coordinates.
More general finite point groups need not embed into a same-dimensional open-boundary Euclidean region; however, prior work shows that orbit spaces of finite group actions admit embeddings into some higher-dimensional Euclidean space with bounded metric distortion~\cite{cahill2025group}.
Thus general space-group codes also admit a natural orbifold locality interpretation.

We now give a practical locality criterion for searching for more local reflection codes on a two-dimensional folded lattice.
For a finite code, we place data qubits on the layer $z=0$ and check qubits on the layer $z=1$.
We define the folded locality as the maximum three-dimensional Euclidean length among all check-data coupling edges:
\begin{equation}
	R_{\rm fold}=\max_{e\in E(H_X)\cup E(H_Z)}\|p_C(e)-p_D(e)\|_2.
\end{equation}
Here $E(H_X)\cup E(H_Z)$ denotes the set of nonzero check-data couplings in the final parity-check matrices.
Figure~\ref{fig:embedding}(c) shows an example of a $[[36,4,6]]$ reflection code with stabilizer terms
\begin{equation}
	\begin{cases}
		f=t_xs_y+s_xs_yt_y+s_xt_xs_yt_y,\\
		h=t_x+s_xt_y^5+s_xt_xt_y .
	\end{cases}
\end{equation}
A folded supercell contains one, two, or four original cells, depending on whether it lies at the origin, on a mirror axis, or in the interior.
Each cell contains two data qubits and two check qubits, placed on two layers for syndrome extraction.
For simplicity, we do not consider the details inside the cell.
The highlighted edge is one of the longest check-data couplings, and its length is the metric of the code locality.

The key advantage of reflection codes over ordinary BB codes is that $s_x$ and $s_y$ are local nearest-neighbor or intra-supercell operations in the folded lattice.
Thus the monomial search space of reflection codes contains many folded-local reflection-assisted terms.
This enlarges the geometric freedom: at the same stabilizer weight, reflection codes have more ways to place algebraic couplings within a short folded-lattice distance.
Consequently, without degrading the code parameters, reflection codes are more likely to yield codes with smaller folded locality.

\subsection{A guiding example}
\label{subsec:guiding-example}

We now study a concrete two-dimensional example by setting
$D=2$ and taking $P$ to be $C_{2v}$, 
generated by two reflections $s_x$ and $s_y$.
Let $t_x$ and $t_y$ denote the translation generators along the $x$ and $y$ directions, 
respectively, so that
$
	G=\langle t_x,s_x,t_y,s_y\rangle .
$
We choose 
\begin{equation}
	\begin{split}
\Gamma_X&=\left(f_X \mid h_X\right)  \\
&=\left(t_x^{-1}s_x t_y^3s_y+t_xs_xt_y^{-1}s_y \mid t_x^{-1}s_xt_y^2s_y+t_xs_xs_y \right)  \\ 
\Gamma_Z &= \left(f_Z \mid h_Z \right)  \\
&= \left( s_yt_y^{-2}s_xt_x+s_ys_xt_x^{-1} \mid s_yt_y^{-3}s_xt_x+s_yt_ys_xt_x^{-1} \right) 
	\end{split}
\end{equation}
to construct a CSS code together with their transposes $\Gamma_X^T$ and $\Gamma_Z^T$. 
One can easily check that 
this is a valid CSS code based on Theorem~\ref{thm:faithful-equivalence}.

For the BB codes or their generalized version, $\Gamma_X^T$, $\Gamma_Z^T$, $\Gamma_X$, 
and $\Gamma_Z$ involve only translation operators and are therefore module homomorphisms.
Consequently, one can define the following chain complex for modules  
\begin{equation}
	\begin{tikzcd}[column sep=large,row sep=large]
		\mathcal R_0
		\arrow[r,"\Gamma_X^T"]
		&
		\mathcal R_0^2
		\arrow[r,"\Gamma_Z"]
		&
		\mathcal R_0
	\end{tikzcd}.
	\label{eq:BB-complex}
\end{equation}
This complex then provides a framework for analyzing the topological 
conditions that characterize the resulting family of BB codes.

A key distinction from the BB-code setting is that our construction 
incorporates point group operations, such as reflections,
in addition to translations.
This raises the question of whether the chain complex in 
Eq.~(\ref{eq:BB-complex}) can still be applied directly.
The answer is negative: generically, $\Gamma_X^T$, $\Gamma_Z^T$, $\Gamma_X$, 
and $\Gamma_Z$ are not module homomorphisms. 
Indeed, consider the reflection $s_x$, acting on the Laurent polynomial ring.
We have $s_x(x\cdot x^2)=x^{-3}$, whereas $x\,s_x(x^2)=x^{-1}$,
confirming that the reflection action fails to preserve the module structure.

To resolve this issue, we introduce a lifted module description over 
an invariant coefficient ring, which we illustrate using the preceding example.
We first define an invariant subring $\mathcal R= \F[u,v]$ of the two 
reflections, where 
\begin{equation} \label{eq:uv-invariants}
	u=x+x^{-1},\qquad v=y+y^{-1},
\end{equation}
are the invariant polynomials under the two reflections. 
Higher powers of $x$ and $y$ reduce to zeroth or first powers 
with coefficients in $\mathcal R$.
Consequently, every $f\in\mathcal R_0$ admits a decomposition of the form
\begin{equation}
f=f_1(u,v)+f_x(u,v)x+f_y(u,v)y+f_{xy}(u,v)xy.
\label{eq:reflection-decomposition}
\end{equation}
This decomposition is unique, as shown by the normal-form computation below.
Hence $\mathcal R_0$ is a free $\mathcal R$-module of rank four, with
basis $\{1,x,y,xy\}$.
Since $u$ and $v$ are invariant under $s_x$ and $s_y$, 
the reflections act as $\mathcal R$-module homomorphisms on $\mathcal R^4$.
With respect to the basis $\{1,x,y,xy\}$, their actions are given by
\begin{equation}
s_x\mapsto
\begin{pmatrix}
1&u&0&0\\
0&1&0&0\\
0&0&1&u\\
0&0&0&1
\end{pmatrix},
\qquad
s_y\mapsto
\begin{pmatrix}
1&0&v&0\\
0&1&0&v\\
0&0&1&0\\
0&0&0&1
\end{pmatrix}.
\label{eq:reflection-action-matrices}
\end{equation}
The translations $t_x$ and $t_y$, which act by multiplication by $x$ and $y$, 
respectively, are also $\mathcal R$-module homomorphisms on $\mathcal R^4$.
In the same basis, they are represented by
\begin{equation}
t_x\mapsto
\begin{pmatrix}
0&1&0&0\\
1&u&0&0\\
0&0&0&1\\
0&0&1&u
\end{pmatrix},
\qquad
t_y\mapsto
\begin{pmatrix}
0&0&1&0\\
0&0&0&1\\
1&0&v&0\\
0&1&0&v
\end{pmatrix}.
\label{eq:translation-action-matrices}
\end{equation}
Therefore the group action extends to a matrix map whose values are $\mathcal R$-module homomorphisms:
\begin{equation}
\Lambda:\F[G]\to \operatorname{End}_{\mathcal R}(\mathcal R^4).
\end{equation}
The chosen isomorphism $\mathcal R_0\cong\mathcal R^4$ sends the actions of $s_x,s_y,t_x,t_y$ to the $4\times4$ matrices above.
Thus each group-algebra element can be represented as a $4\times4$ matrix under $\Lambda$ by multiplication and addition of $s_x,s_y,t_x,t_y$.
The previous two-entry data $(f,h)$ are then represented, under the same isomorphism, by the $4\times8$ block map
\begin{equation}
\bigl(\Lambda(f)\mid \Lambda(h)\bigr):
\mathcal R^8\longrightarrow \mathcal R^4.
\label{eq:lambda-check-map}
\end{equation}
For the two pairs above, the lifted block maps of $\Gamma_X$ and $\Gamma_Z$ are
\begin{widetext}
\begin{equation}
\begin{aligned}
\Lambda_X&=
\left(
\begin{array}{cccc|cccc}
	uv&u^2v&u&u^2+v^2&u&u^2&0&v\\
	0&uv&v^2&u&0&u&v&0\\
	u(v^2+1)&u^2v^2+u^2+v^2&uv&u^2v&uv&u^2v+v&u&u^2\\
	v^2&u(v^2+1)&0&uv&v&uv&0&u
\end{array}
\right),\\
\Lambda_Z&=
\left(
\begin{array}{cccc|cccc}
	u&u^2&0&v&uv&u^2v&u&u^2+v^2\\
	0&u&v&0&0&uv&v^2&u\\
	uv&u^2v+v&u&u^2&u(v^2+1)&u^2v^2+u^2+v^2&uv&u^2v\\
	v&uv&0&u&v^2&u(v^2+1)&0&uv
\end{array}
\right),
\end{aligned}
\label{eq:lambda-mz-example}
\end{equation}
\end{widetext}
respectively.
This example illustrates the module-lifting procedure used below.

\subsection{Module lifting}
\label{subsec:module-lifting}

We now generalize the algebraic module-lifting step illustrated above.
The goal is to find an invariant polynomial subring $\mathcal R$, express $\mathcal R_0$ as a finite free $\mathcal R$-module, and represent all generators of $G$ by $\mathcal R$-module homomorphisms of $\mathcal R^n$.
Here $n$ is the rank of $\mathcal R_0$ as a free $\mathcal R$-module.
After such a basis is fixed, each group-algebra element acts on $\mathcal R^n$ through an $\mathcal R$-module homomorphism:
\begin{equation}
	\Lambda:\F[G]\to \operatorname{End}_{\mathcal R}(\mathcal R^n).
	\label{eq:lambda-action}
\end{equation}

\begin{lemma}
	\label{thm:invariants-subring}
	The invariant set defined as
	$
		\mathcal R_0^P=\{p\in \mathcal R_0\mid g(p)=p,\ \forall g\in P\}
	$
	is a subring of $\mathcal R_0$.
\end{lemma}

\begin{proof}
	For any $p_1,p_2\in\mathcal R_0^P$, we write them as
	$p_1=\sum_m c_m x^m$ and $
	p_2=\sum_n d_n x^n
	$, where $c_m,d_n\in \F$.
	Acting any element $g\in P$ on $p_1 p_2$ leads to
		\begin{align}
		g(p_1p_2)=\sum_{m,n}c_m d_n x^{\rho_g(m+n)}
		=g(p_1)g(p_2)=p_1 p_2,
	\end{align}
   where we have used the fact that $g(x^m)=x^{\rho_g m}$ 
   (see Eq.~(\ref{eq:point-group-action-r0})). In addition, we have
   \begin{equation}
   	g(p_1+p_2)=g(p_1)+g(p_2)=p_1+p_2.
   \end{equation}
   They tell us that $p_1 p_2 \in \mathcal R_0^P$ and $p_1+p_2 \in \mathcal R_0^P$, 
   indicating that $\mathcal R_0^P$ is a subring of $\mathcal R_0$.   
\end{proof}

The invariant ring $\mathcal R_0^P$ needs not itself be a polynomial ring~\cite{derksen2015computational}.
We therefore construct a polynomial subring
\begin{equation}
	\mathcal R\subseteq \mathcal R_0^P\subseteq \mathcal R_0
\end{equation}
such that $\mathcal R_0$ is a finite free $\mathcal R$-module.
This formalism can simplify problems.

\begin{lemma}
	\label{lem:finite-module-integral}
	Let $B=A[b_1,\ldots,b_N]$ be a finitely generated $A$-algebra.
	If each $b_i$ is integral over $A$, then $B$ is a finite $A$-module.
\end{lemma}

\begin{proof}
	Since $b_i$ is integral over $A$, there exists a positive integer 
	$m$ and $a_j\in A$ such that
	\begin{equation}
		b_i^m+a_{m-1}b_i^{m-1}+\cdots+a_0=0,\qquad a_j\in A.
	\end{equation}
	Thus, $b_i^m$ is an $A$-linear combination of lower powers, 
	and all higher powers of $b_i$ reduce to the finite set $\{1,b_i,\ldots,b_i^{m-1}\}$.
	Hence $A[b_i]$ is a finite $A$-module.
	The case of several generators follows straightforwardly.
\end{proof}

\begin{theorem}
	\label{thm:r0-finite-over-invariants}
	$\mathcal R_0$ is a finite $\mathcal R_0^P$-module.
\end{theorem}

\begin{proof}
	For any $p\in\mathcal R_0$, let $P=\{g_1,g_2,\dots,g_n\}$, where $n=|P|$,
	and set $a_i=g_i(p)$.
	Consider the orbit polynomial
	\begin{equation}
		F_p(T)=\prod_{i=1}^n (T+a_i)=T^n + u_1 T^{n-1}+\dots+u_n,
	\end{equation}
    where $T\in R_0$ and $
    u_j = \sum_{1\le i_1<\cdots<i_j\le n} a_{i_1}\cdots a_{i_j}
    $. In particular, $u_1=\sum_{g\in P}g(p)$. 
    
    We claim that each coefficient $u_j$ is invariant under the action of $P$. For $u_1$,
    this is immediate from the fact that left multiplication by any $g'\in P$ permutes 
    the set $P$, yielding $g'(u_1 )=u_1$. 
    For general $j$, the coefficient $u_j$ is the sum of all products of the form
    $
    g_{k_1}(p)g_{k_2}(p)\cdots g_{k_j}(p),
    $
    taken over all $j$-element subsets $\{g_{k_1},\dots,g_{k_j}\}\subset P$ with distinct indices.
    Applying $g'$ to such a term leads to $g'(g_{k_1}(p) g_{k_2}(p)\cdots g_{k_j}(p))
    =(g' g_{k_1})(p) (g' g_{k_2})(p)\dots (g' g_{k_j})(p)=g_{k_1'}(p) g_{k_2'}(p)…g_{k_j'}(p)$ 
    where $g' g_{k_i}=g_{k_{i'}}$. 
    For two disjoint $j$-element subsets $S_1=\{g_{k_1},\dots,g_{k_j}\}$ and $S_2=\{g_{l_1},\dots,g_{l_j}\}$ 
    of $P$, we claim that their images $g'S_1$ and $g'S_2$ remain disjoint. Indeed, if 
    $g'S_1\cap g' S_2\neq\varnothing$, then $g'g_{k_a}=g'g_{l_b}$ for some indices $a,b$, 
    and left cancellation gives $g_{k_a}=g_{l_b}$, contradicting the 
    assumption that $S_1\cap S_2=\varnothing$. Therefore $g'(u_j)=u_j$, so $u_j\in\mathcal R_0^P$.
    
	The polynomial is monic and has $p$ as a root, that is, $F_p(p)=0$.
	Thus every $p\in\mathcal R_0$ is integral over $\mathcal R_0^P$.
	Since $P$ is finite, Lemma~\ref{lem:finite-module-integral} implies that $\mathcal R_0$ is a finite $\mathcal R_0^P$-module.
\end{proof}

\begin{theorem}
	\label{thm:polynomial-subring-free}
	There exists a polynomial subring
	\begin{equation}
		\mathcal R=\F[\phi_1,\ldots,\phi_D]\subseteq\mathcal R_0^P
	\end{equation}
	such that $\mathcal R_0$ is a finite free $\mathcal R$-module.
\end{theorem}

\begin{proof}
	The argument uses standard finite-group invariant theory and commutative algebra~\cite{derksen2015computational,eisenbud1995commutative}.
	The ring $\mathcal R_0$ is a finite $\mathcal R_0^P$-module.
	The Artin--Tate lemma then implies that $\mathcal R_0^P$ is a finitely generated $\F$-algebra~\cite{eisenbud1995commutative}.
	Since $\mathcal R_0$ has Krull dimension $D$ and integral extensions preserve Krull dimension, $\mathcal R_0^P$ also has Krull dimension $D$~\cite{eisenbud1995commutative}.
	Noether normalization applied to $\mathcal R_0^P$ therefore gives $D$ algebraically independent invariants $\phi_1,\ldots,\phi_D$ such that $\mathcal R_0^P$ is finite over $\mathcal R=\F[\phi_1,\ldots,\phi_D]$~\cite{eisenbud1995commutative}.
	Since $\mathcal R_0$ is finite over $\mathcal R_0^P$, it is finite over $\mathcal R$.
	The remaining step is a standard commutative-algebra freeness statement.
	The Laurent polynomial ring $\mathcal R_0$ is regular and hence Cohen--Macaulay, while $\mathcal R$ is a polynomial ring, which is regular~\cite{eisenbud1995commutative}.
	By the local miracle-flatness criterion, equivalently the finite Cohen--Macaulay-over-regular case, $\mathcal R_0$ is finite flat over $\mathcal R$~\cite{eisenbud1995commutative}.
	Since $\mathcal R$ is Noetherian and $\mathcal R_0$ is finite over $\mathcal R$, the module $\mathcal R_0$ is finitely presented; hence this finite flat module is finite projective~\cite{eisenbud1995commutative}.
	By the Quillen--Suslin theorem, every finite projective module over a polynomial ring is free~\cite{lam2006serre}.
	Therefore there is a finite integer $n$ and an $\mathcal R$-basis
	\begin{equation}
		E=\{e_1,\ldots,e_n\}
	\end{equation}
	such that
	\begin{equation}
		\mathcal R_0\cong \mathcal R^n.
	\end{equation}
\end{proof}

Theorem~\ref{thm:polynomial-subring-free} is an existence statement.
For the examples used in this paper, we choose $\mathcal R$ from explicit orbit sums and verify the resulting finite-free presentation by a Gr\"{o}bner normal-form calculation.
For each orbit
\begin{equation}
	O=\{g(p)\mid g\in P\},
\end{equation}
define
\begin{equation}
	s_O=\sum_{\chi\in O}\chi.
\end{equation}
Although there are infinitely many such orbit sums, the concrete two-dimensional point groups in this work use the verified choices listed in Appendix~\ref{app:basis-table}.
For a new point group, the same verification consists of choosing candidate invariants, computing the normal-form basis below, and checking that the listed monomials span uniquely over the chosen polynomial subring.

The basis of $\mathcal R_0$ as a module over $\mathcal R$ is computed by a Gr\"{o}bner normal form~\cite{cox1997ideals,greuel2008singular}.
Rewrite the Laurent polynomial ring as a quotient of an ordinary polynomial ring by introducing inverse variables $X_i,Y_i$ with $Y_i=X_i^{-1}$.
Introduce base variables $U_1,\ldots,U_D$ for the invariants $\phi_1,\ldots,\phi_D$ and define
\begin{equation}
	B=\F[U_1,\ldots,U_D,X_1,\ldots,X_D,Y_1,\ldots,Y_D].
\end{equation}
Let
\begin{equation}
	I=\langle X_iY_i+1,\ U_j+\Phi_j(X,Y)\rangle,
	\label{eq:normal-form-ideal}
\end{equation}
where $\Phi_j(X,Y)$ is the expression of $\phi_j$ in the variables $X_i,Y_i$.
A Gr\"{o}bner basis of $I$ gives normal-form monomials~\cite{cox1997ideals,greuel2008singular}.
Those not divisible by leading terms form an $\mathcal R$-basis.
For the reflection example, where $Y_1=X_1^{-1}$ and $Y_2=X_2^{-1}$,
\begin{equation}
	X_1^2+U X_1+1=0,\qquad X_2^2+V X_2+1=0.
\end{equation}
The leading terms are $X_1^2,X_2^2,Y_1,Y_2$, so the standard monomials are
\begin{equation}
	\{1,X_1,X_2,X_1X_2\},
\end{equation}
which correspond to $\{1,x,y,xy\}$ in the original Laurent ring.
This confirms the uniqueness of Eq.~\eqref{eq:reflection-decomposition}.

Let
\begin{equation}
	\theta:\mathcal R^n\xrightarrow{\cong}\mathcal R_0
\end{equation}
be the chosen free-module identification.
The lifted matrix map $\Lambda$ is defined by transporting the original group-algebra action $\Gamma$ through this basis:
\begin{equation}
	\Lambda(q)=\theta^{-1}\Gamma(q)\theta,\qquad q\in\F[G].
	\label{eq:lambda-gamma-conjugacy}
\end{equation}
Equivalently, $\theta\Lambda(q)=\Gamma(q)\theta$.
For the chosen pairs $(f_X,h_X)$ and $(f_Z,h_Z)$, define
\begin{equation} \label{eq:lambda-syndrome-maps}
	\begin{split}
		\Lambda_X &= \bigl(\Lambda(f_X)\mid \Lambda(h_X)\bigr),\\
		\Lambda_Z &= \bigl(\Lambda(f_Z)\mid \Lambda(h_Z)\bigr).
	\end{split}
\end{equation}
$\Lambda_X$ and $\Lambda_Z$ are exactly the $\mathcal R$-basis 
expressions of the original check maps $H_X$ and $H_Z$ in Eqs.~\eqref{eq:gamma-block-maps}.

The transpose maps are defined by applying the same group-algebra transpose $q\mapsto q^T$ before forming the column block:
\begin{equation} \label{eq:lambda-assembled-check}
	\begin{split}
	\Lambda_X^T &=\bigl(\Lambda(f_X^T)\mid \Lambda(h_X^T)\bigr)^T,\\
	\Lambda_Z^T &=\bigl(\Lambda(f_Z^T)\mid \Lambda(h_Z^T)\bigr)^T.
	\end{split}
\end{equation}
Here $\Lambda_X^T,\Lambda_Z^T:\mathcal R^n\to\mathcal R^{2n}$.
The compatibility with the ordinary transpose of the original permutation check matrices is obtained by applying Eq.~\eqref{eq:lambda-gamma-conjugacy} to $f_X^T,h_X^T$ and $f_Z^T,h_Z^T$.
Thus $\Lambda_X^T$ and $\Lambda_Z^T$ are the $\mathcal R$-basis expressions of $H_X^T$ and $H_Z^T$, respectively.
They should not be confused with the ordinary transposes of the $\mathcal R$-matrices $\Lambda_X$ and $\Lambda_Z$; after changing from the Laurent-monomial basis to a finite $\mathcal R$-basis, formal group-algebra transpose and ordinary matrix transpose need not coincide.

\subsection{Topological-order condition}
\label{subsec:topological-order-condition}

With the block-map notation fixed in Eqs.~\eqref{eq:gamma-block-maps},
~\eqref{eq:lambda-syndrome-maps}, and~\eqref{eq:lambda-assembled-check}, the lifted CSS complex is
\begin{equation}
	\begin{tikzcd}[column sep=large,row sep=large]
		\mathcal R_0
		\arrow[r,"\Gamma_X^T"]
		&
		\mathcal R_0^2
		\arrow[r,"\Gamma_Z"]
		&
		\mathcal R_0
		\\
		\mathcal R^n
		\arrow[u]
		\arrow[r,"\Lambda_X^T"']
		&
		\mathcal R^{2n}
		\arrow[u]
		\arrow[r,"\Lambda_Z"']
		&
		\mathcal R^n
		\arrow[u] .
	\end{tikzcd}
	\label{eq:lifted-css-complex}
\end{equation}
The dual lifted complex is obtained by exchanging $X$ and $Z$ and applying the same formal transpose.

We use the Buchsbaum--Eisenbud criterion to test exactness 
at the middle Pauli-support module~\cite{buchsbaum1975generic}.
We state this criterion as a lemma below and refer the interested 
readers to Ref.~\cite{buchsbaum1975generic} for further details.
\begin{lemma}[Buchsbaum--Eisenbud criterion]
	\label{lem:buchsbaum-eisenbud}
	Let $\mathcal R$ be a Noetherian ring and let
	\begin{equation}
		C_\bullet:\quad
		0\to C_m\xrightarrow{\partial_m}C_{m-1}
		\xrightarrow{\partial_{m-1}}\cdots
		\xrightarrow{\partial_2}C_1\xrightarrow{\partial_1}C_0
	\end{equation}
	be a finite free complex.
	Let $r_i=\rank \partial_i$ and let $I_{r_i}(\partial_i)$ be the ideal generated 
	by all $r_i\times r_i$ minors of $\partial_i$.
	Then $C_\bullet$ is exact if and only if, for every $i$,
	\begin{equation}
		\rank C_i=r_i+r_{i+1},\qquad
		\grade I_{r_i}(\partial_i)\ge i.
		\label{eq:be-condition}
	\end{equation}
\end{lemma}

The algebraic topological-order target used here is middle exactness of the lifted CSS complex and middle exactness of its dual:
\begin{equation}
	\imop(\Lambda_X^T)=\kerop(\Lambda_Z),
	\qquad
	\imop(\Lambda_Z^T)=\kerop(\Lambda_X).
\end{equation}
For the lifted CSS complex
\begin{equation}
	0\to \mathcal R^n\xrightarrow{\Lambda_X^T}
	\mathcal R^{2n}\xrightarrow{\Lambda_Z}\mathcal R^n,
\end{equation}
the Buchsbaum--Eisenbud condition for middle exactness gives both the middle rank condition and the codimension-two condition for the incoming map:
\begin{equation} 
	\begin{split}
	&\rank(\Lambda_X^T)+\rank(\Lambda_Z)=2n,\\
	&\grade I_{\rank(\Lambda_X^T)}(\Lambda_X^T)\ge2.
	\end{split}
\end{equation}
Since both maps have rank at most $n$, this is equivalent to $\rank(\Lambda_X^T)=\rank(\Lambda_Z)=n$ and $\grade I_n(\Lambda_X^T)\ge2$.
The Buchsbaum--Eisenbud grade condition for the outgoing map $\Lambda_Z$ has index one and is automatic once $\Lambda_Z$ has full rank over the polynomial domain.

Similarly, for the dual lifted complex
\begin{equation}
	0\to \mathcal R^n\xrightarrow{\Lambda_Z^T}
	\mathcal R^{2n}\xrightarrow{\Lambda_X}\mathcal R^n,
\end{equation}
the middle exactness gives
\begin{equation} 
	\begin{split}
	&\rank(\Lambda_Z^T)+\rank(\Lambda_X)=2n,\\
	&\grade I_{\rank(\Lambda_Z^T)}(\Lambda_Z^T)\ge2,
	\end{split}
\end{equation}
or equivalently $\rank(\Lambda_Z^T)=\rank(\Lambda_X)=n$ and $\grade I_n(\Lambda_Z^T)\ge2$.
Therefore, for the length-two lifted free complexes considered here, the following rank-and-grade conditions are necessary and sufficient for the middle exactness target above and its dual:
\begin{equation}
	\begin{cases}
		\rank(\Lambda_X^T)=\rank(\Lambda_Z)=n,\\
		\grade I_{n}(\Lambda_X^T)\ge 2,\\
		\rank(\Lambda_Z^T)=\rank(\Lambda_X)=n,\\
		\grade I_{n}(\Lambda_Z^T)\ge 2.
	\end{cases}
	\label{eq:topological-order-condition}
\end{equation}
The rank and grade tests in Eq.~\eqref{eq:topological-order-condition} are therefore performed on the formal transpose maps that actually enter the CSS complex.
They cannot in general be replaced by ordinary transpose invariance of the $\mathcal R$-matrices unless an additional unimodular pairing or an equivalent left-right equivalence has been supplied.
The syndrome quotient modules used below, such as $\coker(\Lambda_X)$ and $\coker(\Lambda_Z)$, are therefore compatible with this middle-exactness criterion; they record sector data associated with finite-support Pauli boundaries.

Each grade condition in the criterion can be converted into a gcd test.
\begin{theorem}
	\label{thm:gcd-grade-equivalence}
	Let $F$ be either $\Lambda_X^T$ or $\Lambda_Z^T$, and assume $\rank(F)=n$.
	Then
	\begin{equation}
		\grade I_n(F)\ge 2
	\end{equation}
	is equivalent to the condition that the $n\times n$ minors of $F$ have no nontrivial common factor:
	\begin{equation}
		\gcd I_n(F)=1.
	\end{equation}
\end{theorem}

\begin{proof}
	The ring $\mathcal R=\F[\phi_1,\ldots,\phi_D]$ is a polynomial ring over a field.
	We recall the needed algebraic properties.
	A ring is Noetherian if all of its ideals are finitely generated; by Hilbert's basis theorem, polynomial rings over a field are Noetherian~\cite{eisenbud1995commutative}.
	A ring is regular if each localization at a prime ideal is a regular local ring, meaning that the maximal ideal can be generated by exactly the Krull dimension many elements~\cite{eisenbud1995commutative}.
	Polynomial rings over a field are regular, for example because they are smooth affine algebras over a field~\cite{eisenbud1995commutative}.
	A Cohen--Macaulay ring is one in which depth equals Krull dimension after localization at every prime, and every regular local ring is Cohen--Macaulay~\cite{eisenbud1995commutative}.
	Thus $\mathcal R$ is Cohen--Macaulay.
	Finally, a unique factorization domain (UFD) is a domain with 
	unique factorization into irreducibles, and $\mathcal R$ is a UFD by Gauss's lemma since it is a polynomial ring over the field $\F$~\cite{eisenbud1995commutative}.
	Let $I=I_n(F)$.
	Since $\rank(F)=n$, at least one $n\times n$ minor is nonzero, so $I\ne0$.
	In a Cohen--Macaulay ring, $\grade I=\operatorname{height} I$~\cite{eisenbud1995commutative}.
	Thus $\grade I\ge2$ is equivalent to saying that $I$ is contained in no height-one prime ideal.
	Because $\mathcal R$ is a UFD, every height-one prime ideal is generated by an irreducible polynomial $q$~\cite{eisenbud1995commutative}.
	If all generators of $I$ share such a factor $q$, then $I\subseteq(q)$ and $\operatorname{height}I\le1$.
	Conversely, if $\operatorname{height}I<2$, then $I$ lies in a height-one prime $(q)$, so $q$ divides every $n\times n$ minor.
	This proves the equivalence.
\end{proof}

\subsection{Anyon counting}
\label{subsec:anyon-counting}

\subsubsection{Infinite lattices}
\label{subsubsec:infinite-lattice-sectors}

The quotient modules below record equivalence classes in the syndrome module.
Two syndromes are equivalent if they differ by the boundary of a finite-support Pauli operator.
For standard local translation-invariant stabilizer models, these quotient data 
count independent anyon types~\cite{haah2013commuting,liang2025generalized}.
For space-group codes, we adopt the lifted module representation, in which
the syndrome maps are given by $\Lambda_X,\Lambda_Z:\mathcal R^{2n}\to\mathcal R^n$.
The anyons are then classified by the cokernels:
\begin{equation}
	\coker(\Lambda_X),
	\qquad
	\coker(\Lambda_Z).
	\label{eq:coker-sectors}
\end{equation}
If both cokernels have finite dimension as $\F$-vector spaces, 
the total number of independent anyon species---which 
also equals the number of logical qubits---is
\begin{equation}
	k=\dim_{\F}\coker(\Lambda_X)+
	\dim_{\F}\coker(\Lambda_Z).
	\label{eq:sector-number}
\end{equation}
Note that these two vector-space dimensions are algebraic dimensions of polynomial-module quotients.
They are computed by a module Gr\"{o}bner basis: after reducing by the submodule generated 
by the columns of $\Lambda_X$ or $\Lambda_Z$, the remaining standard module monomials form an $\F$-basis~\cite{greuel2008singular}.
Appendix~\ref{app:groebner-quotients} gives the quotient-ring version and its module extension.

\subsubsection{Finite periodic lattices}
\label{subsubsec:finite-periodic-logical-qubits}

This subsection explains how the module calculation on the infinite lattice descends to a finite periodic system.  
Algebraically, finite size is imposed by replacing the Laurent ring $\mathcal R_0$ with the quotient
\begin{equation}
	\mathcal R_{0,L}=\mathcal R_0/I_L ,
	\label{eq:finite-periodic-ring}
\end{equation}
where $I_L$ is generated by the periodic boundary relations.  
For example, an $L_x\times L_y$ periodic system has
\begin{equation}
	I_L=\langle x^{L_x}+1,\ y^{L_y}+1\rangle .
	\label{eq:finite-periodic-ideal}
\end{equation}

We first record the descent condition for the point-group action.  
Let $g\in P$ act on $\mathcal R_0$.  
If it descends to the quotient $\mathcal R_{0,L}$, the descended map must have the form
\begin{equation}
	\bar g([f])=[g(f)],
	\label{eq:descended-point-action}
\end{equation}
where $[f]$ denotes the class of $f$ modulo $I_L$.

\begin{proposition}[Finite descent condition]
	\label{thm:finite-descent-condition}
	The formula in Eq.~\eqref{eq:descended-point-action} is well defined if and only if
	\begin{equation}
		g(I_L)\subseteq I_L .
		\label{eq:finite-period-stability}
	\end{equation}
\end{proposition}

\begin{proof}
	Let $f,f'\in\mathcal R_0$.  
	The equality $[f]=[f']$ is equivalent to $f-f'\in I_L$.  
	The formula in Eq.~\eqref{eq:descended-point-action} is well defined precisely when $[g(f)]=[g(f')]$ for every such pair.  
	This is equivalent to
	\begin{equation}
		g(f)-g(f')\in I_L.
	\end{equation}
	Since $g$ is a ring homomorphism,
	\begin{equation}
		g(f)-g(f')=g(f-f').
	\end{equation}
	Thus the condition is that $g(h)\in I_L$ for every $h\in I_L$, which is exactly Eq.~\eqref{eq:finite-period-stability}.  
\end{proof}

We now express the period relations in the folded module coordinates.  
Let
\begin{equation}
	\theta:\mathcal R^n\xrightarrow{\cong}\mathcal R_0
\end{equation}
be the fixed $\mathcal R$-module identification.  
The pullback of the period ideal is the submodule
\begin{equation}
	\theta^{-1}(I_L)\subseteq\mathcal R^n .
\end{equation}
Choose a presentation matrix $M_L$ such that
\begin{equation}
	\operatorname{im}M_L=\theta^{-1}(I_L).
	\label{eq:finite-period-module}
\end{equation}
Then the finite folded module is
\begin{equation}
	\mathcal R_L^n=\mathcal R^n/\operatorname{im}M_L \cong\mathcal R_{0,L}.
	\label{eq:finite-folded-module}
\end{equation}
The matrix $M_L$ can be written directly from the generators of the period ideal:
\begin{equation}
	M_L=
	\bigl(
	\Lambda(t_1^{L_1}+1)\mid
	\cdots\mid
	\Lambda(t_D^{L_D}+1)
	\bigr).
	\label{eq:period-presentation-matrix}
\end{equation}
The infinite check maps
$
	\Lambda_X$ and $\Lambda_Z
$
descend to finite maps
\begin{equation}
	\Lambda_{X,L},\Lambda_{Z,L}:(\mathcal R_L^n)^2\to\mathcal R_L^n,
\end{equation}
when they are compatible with the period relations.

\begin{theorem}[Finite periodic logical count]
	\label{thm:finite-periodic-logical-count}
	Assume that the finite maps $\Lambda_{X,L}$ and $\Lambda_{Z,L}$ are well defined.  
	Then the number of logical qubits is
	\begin{equation}
		k_L=\dim_{\F}\frac{\mathcal R^n}{\operatorname{im}M_L+\operatorname{im}\Lambda_X}+\dim_{\F}\frac{\mathcal R^n}{\operatorname{im}M_L+\operatorname{im}\Lambda_Z}.
		\label{eq:finite-periodic-logical-count}
	\end{equation}
\end{theorem}

\begin{proof}
    The number of logical qubits is given by
	\begin{equation}
		k_L=\dim_{\F}\frac{\mathcal R_L^n}{\operatorname{im}\Lambda_{X,L}}+\dim_{\F}\frac{\mathcal R_L^n}{\operatorname{im}\Lambda_{Z,L}}.
		\label{eq:finite-logical-cokernel}
	\end{equation} 
	By the submodule correspondence theorem for quotient modules~\cite{atiyah2018introduction}, we have
	\begin{equation}
		\operatorname{im}\Lambda_{X,L}=\frac{\operatorname{im}\Lambda_X+\operatorname{im}M_L}{\operatorname{im}M_L}.
		\label{eq:finite-x-image}
	\end{equation}
	Hence
	\begin{equation}
		\frac{\mathcal R_L^n}{\operatorname{im}\Lambda_{X,L}}=\frac{\mathcal R^n/\operatorname{im}M_L}{(\operatorname{im}\Lambda_X+\operatorname{im}M_L)/\operatorname{im}M_L}\cong
		\frac{\mathcal R^n}{\operatorname{im}M_L+\operatorname{im}\Lambda_X},
		\label{eq:finite-x-quotient}
	\end{equation}
    and similarly for the $Z$ case.
    In the last step of derivation, we use the third isomorphism theorem~\cite{atiyah2018introduction}.
    Eq.~(\ref{eq:finite-periodic-logical-count}) thus holds.
\end{proof}

Equation~\eqref{eq:finite-periodic-logical-count} is the form used for the module Gr\"{o}bner-basis computation.  
The periodic relations and the check image generate a single submodule of $\mathcal R^n$, and the remaining standard module monomials give the required $\F$-dimension.

\subsection{Space-group symmetries and logical gates by relabeling}
\label{subsec:sg-symmetry-logical-gates}

In this subsection, we show how
translation symmetry is violated in space-group codes and discuss 
non-translation symmetries possessed by these codes 
and logical operations induced by such symmetries.

To discuss symmetry in a stabilizer code, we view a code 
as a commuting stabilizer Hamiltonian by assigning each Hamiltonian term 
to each chosen stabilizer generator with the coefficient being $-1$.  
For a single space-group template of the form
\begin{equation}\label{eq:fh}
(f^T \mid h^T),\quad (h \mid f),
\end{equation}
the $X$-type and $Z$-type stabilizer generators are given by
\begin{equation}\label{eq:stabilizergenerators}
(fp\mid hp),\quad (h^T p\mid f^Tp),
\end{equation}
where $p$ ranges over the monomials of $\mathcal{R}_0$. One may notice that 
compared with Eq.~(\ref{eq:fh}), there is an extra transpose operation in 
Eq.~(\ref{eq:stabilizergenerators}). This is because we adopt the 
standard matrix representation defined in Eq.~(\ref{eq:permutation-matrix-entry}), 
whereas other studies define the matrix representation using the transposed form. 

A space-group element $g\in G$ is a stabilizer Hamiltonian symmetry 
when it permutes the Hamiltonian terms. Equivalently, the transformed 
generator multiset must be the original generator multiset, that is,
\begin{equation}
\begin{split}
&\{(gfp\mid g hp):p\}=\{(fp\mid hp):p\}, \\
&\{(gh^T p\mid gf^T p):p\}=\{(h^T p\mid f^T p):p\}.
\end{split}
\end{equation}
A convenient algebraic test is the existence of a common right shift $r_X\in G$
and $r_Z\in G$ such that
\begin{equation}
	\begin{split}
gfg^{-1}=fr_X,&\quad ghg^{-1}=hr_X, \\
gh^Tg^{-1}=h^Tr_Z,&\quad gf^Tg^{-1}=f^Tr_Z.
\end{split}
\end{equation}
The special case $r_X=r_Z=1$ gives the stronger template-level sufficient condition
\begin{equation}
	\begin{split}
gfg^{-1}=f,&\quad ghg^{-1}=h, \\
gf^Tg^{-1}=f^T,&\quad gh^Tg^{-1}=h^T.
\end{split}
\end{equation}

We therefore clearly see that the BB codes and their generalized version
have translation invariance as $f,h,f^T,h^T$ contain only translations.
In contrast, a general space-group code contains point-group operations, 
and translations need not commute with those operations.  
Thus a space-group code is generally not translation invariant, even when it has topological order.  
This broadens the usual translation-invariant viewpoint on topological stabilizer phases.
Moreover, space-group codes can nevertheless have more general symmetries.  
In the example of Subsec.~\ref{subsec:guiding-example}, the non-translation 
element $g=s_xs_yt_y^{-1}$ is a symmetry of the stabilizer Hamiltonian.  
In this example, the action of $g$ conjugates the template 
terms pairwise and leaves the Hamiltonian invariant.

In the example, the code has $k=8$ logical qubits with
$4$ independent $e$-type anyons and $4$ independent $m$-type anyons.  
In the lifted module description, one may choose the $e$-sector basis as
\begin{equation}
	e_1=[1],\quad e_2=[x],\quad e_3=[y],\quad e_4=[xy],
\end{equation}
and similarly
\begin{equation}
	m_1=[1],\quad m_2=[x],\quad m_3=[y],\quad m_4=[xy].
\end{equation}
The symmetry $g=s_xs_yt_y^{-1}$ acts by
\begin{equation}
	e_1\leftrightarrow e_3,\qquad e_2\leftrightarrow e_4,
\end{equation}
and
\begin{equation}
	m_1\leftrightarrow m_3,\qquad m_2\leftrightarrow m_4.
\end{equation}
In matrix form, on each four-dimensional sector quotient,
after choosing a compatible finite logical basis, 
this sector action appears as a SWAP gate on four pairs of logical qubits.

From Subsec.~\ref{subsec:operator-group}, we know that the action of $g$ 
induces a permutation on the physical qubit labels. For example,
$g (xy)=s_xs_yt_y^{-1} (xy)=x^{-1}$, indicating that a qubit initially labeled by $(1,1)$ is sent 
to $(-1,0)$, effecting a relabeling of the corresponding qubits.
Such a permutation can, in principle, be realized by physically moving or 
swapping qubits. Alternatively, it may be viewed purely as a formal relabeling operation, 
in which case no physical quantum operation is required; one simply updates the qubit labels, 
decoder labels, and readout conventions consistently. This relabeling directly implements the 
SWAP gate depicted above. More generally, these permutations provide a mechanism for 
automorphism-induced logical gates, analogous to the construction of logical gates via 
qubit relabeling~\cite{koh2026entangling}.

\section{Folded locality in AOD movement} 
\label{sec:movement-illustration}

\begin{figure*}[t]
\centering
\includegraphics[width=\textwidth]{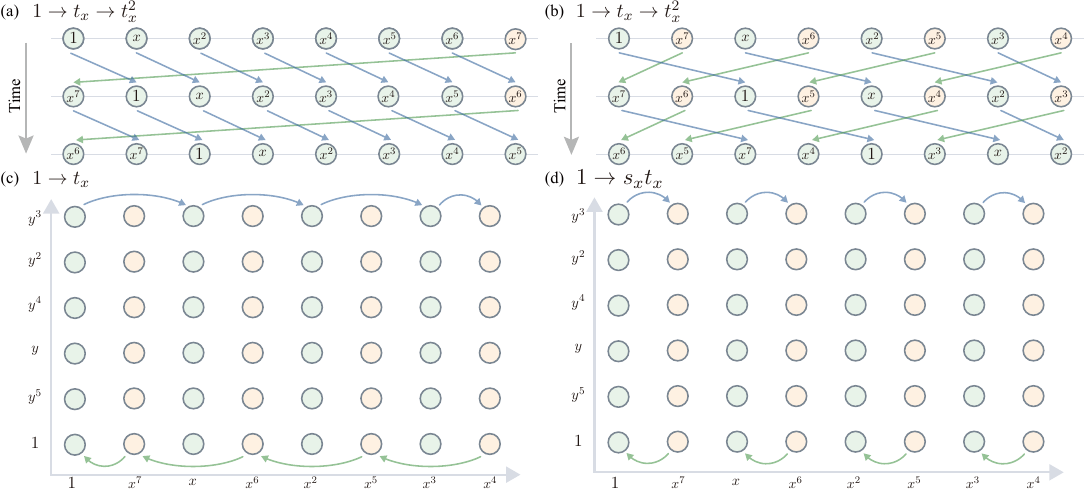}
\caption{
AOD movement protocols from folded locality.
Arrows indicate the induced motion of check qubits between consecutive syndrome-extraction substeps.
Different colored arrows need to be completed in different substeps to avoid collisions.
(a) In a one-dimensional periodic order, the transition $1\to t_x\to t_x^2$ contains a boundary wrap move whose length scales with the system size.
(b) A folded order places boundary labels next to each other, converting the wrap move into bounded-distance AOD substeps.
(c) In two dimensions, folding both coordinates gives a folded $x^a y^b$-labeled check array. 
The $t_x$ transition decomposes into right-moving and left-moving column batches.
(d) $s_x t_x$ corresponds to local swaps of columns of atoms.
}
\label{fig:movement}
\end{figure*}

In Subsec.~\ref{subsec:sg-locality-main}, we discuss the locality of space-group codes from both geometric and algebraic viewpoints, with particular emphasis on the folded locality of reflection codes. 
The key observation there is that, after lattice cells are reorganized 
according to reflection orbits, couplings that appear to cross the periodic boundary in the original lattice can become bounded-range couplings on the folded lattice. 
We now explain how this folded locality can translate into an experimental advantage for syndrome extraction in reconfigurable neutral atom arrays~\cite{bluvstein2024logical,tan2022qubit,wang2024atomique,evered2023high,bluvstein2026fault,radnaev2025universal,reichardt2024fault,cong2022hardware}.

We consider syndrome extraction with three types of physical qubits: data qubits, $X$-check qubits, and $Z$-check qubits. 
The data qubits store the protected quantum information, while the $X$- and $Z$-check qubits serve as ancillas for measuring $X$- and $Z$-type stabilizers. 
On a neutral atom platform, measuring one stabilizer typically requires moving the corresponding check qubit close to each data qubit in the stabilizer support, applying an entangling gate such as a Rydberg-mediated CZ or CNOT gate, and then measuring the check qubit to read out the syndrome. 
Since a full quantum error correction (QEC) round measures many stabilizers, the relevant question is how efficiently many check qubits can be moved in parallel.

An AOD provides a natural mechanism for such parallel motion.
It can translate many optical tweezers simultaneously, for example by moving multiple rows or columns of atoms in the horizontal or vertical direction. 
This parallelism is an important resource for qLDPC code syndrome extraction in neutral atom arrays. 
At the same time, AOD motion is geometrically constrained: simultaneously executed paths should not overlap, cross, or lead to collisions. 
When different moving groups would cross, the motion must be decomposed into several substeps~\cite{xu2024constant,zhao2026towards}.

We use the following coarse-grained movement model. Data qubits are fixed in a middle layer, 
while $X$-check and $Z$-check qubits occupy two movable layers above 
and below the data layer (in practice, these qubits may all be placed within the same layer). 
During syndrome extraction, the check layers are rearranged by AOD motion so that each check qubit visits the cells corresponding to its target data qubits. 
For a given AOD substep, the movement time is controlled by the largest displacement among the rows or columns moved in that substep. 
Parking sites, trap transfer, and pulse-level scheduling are left to a lower-level hardware compiler.

As in the algebraic construction, we label lattice cells by monomials in $\mathcal{R}_0$. 
In one dimension the cells are labeled by $x^a$, while in two dimensions they are labeled by $x^a y^b$.
In the two-block CSS representation, each cell contains one data qubit from each data block, and also has a standard position for one $X$-check qubit and one $Z$-check qubit. 
Since the data layer is fixed, the syndrome-extraction movement problem reduces to a rearrangement problem for the check layers. 
If the $X$-check terms are $f_X=g_1^{-1}+g_2^{-1},\quad h_X=g_3^{-1}+g_4^{-1}$,
where $g_1$, $g_2$, $g_3$, and $g_4$ are elements in a space group, then an $X$-check qubit must visit the corresponding monomial positions in sequence. 
Equivalently, the relative layer transitions are determined by
\begin{equation}
g_1,\quad g_2 g_1^{-1},\quad g_3 g_2^{-1},\quad g_4 g_3^{-1},
\end{equation}
up to the chosen ordering convention for the syndrome-extraction circuit. 
The same discussion applies to $Z$-checks.

Figure~\ref{fig:movement} illustrates the basic mechanism. 
In Fig.~\ref{fig:movement}(a), we take $L=8$ and first consider a one-dimensional translation layout. 
The cells are arranged in the ordinary periodic order $1,\ x,\ x^2,\ldots,\ x^7 $.
For the sequence $1\to t_x\to t_x^2$, most check qubits move by one lattice spacing under each $t_x$ transition. 
However, the qubit at the periodic boundary must move from the right edge back to the left edge. 
Thus every $t_x$ layer contains a wrap-around move whose length scales with $L$.
Although these moves can be parallelized, the longest displacement in the substep grows with the system size, consuming an increasing fraction of the QEC cycle time.

Figure~\ref{fig:movement}(b) shows the same sequence after a folded rearrangement. 
The cells are ordered as $1,\ x^7,\ x,\ x^6,\ x^2,\ x^5,\ x^3,\ x^4$.
This folded order places the two sides of the periodic boundary next to each other. 
As a result, the apparent wrap-around transition is converted into a bounded-distance motion. 
To avoid path crossings, the $t_x$ transition can be split into two muted batches of short moves.
Crucially, after this splitting the maximum displacement in each substep is $O(1)$, independent of $L$.
The same idea extends coordinatewise to two dimensions. 
In Fig.~\ref{fig:movement}(c), cells are labeled by $x^a y^b$, and the $x$-and $y$-coordinates are folded independently. 
A translation $t_x:x^a y^b\mapsto x^{a+1}y^b$ acts columnwise on the folded array. 
Some columns move to the right and are drawn with arrows above the array; the remaining columns move to the left and are drawn with arrows below the array.
Thus the two-dimensional folded layout is the product of two one-dimensional foldings: the periodic cost in the $x$-direction is reduced by folding the $x$-coordinate, and similarly for the $y$-direction.

For ordinary BB codes, the available monomials are pure translations $t_x^a t_y^b$. 
Folded layouts can reduce the movement cost of such translation layers, but they do not add new local algebraic operations to the stabilizer template. 
Reflection codes enlarge this space. 
As shown in Fig.~\ref{fig:movement}(d), a glide reflection term such as $s_x t_x$ 
acts locally on the folded labels: it exchanges nearby folded partners rather than moving a qubit across the full periodic system. 
Therefore reflection and glide reflection terms can be implemented as combinations of folded-local movement primitives. 
This is the experimental counterpart of the code-native folded locality discussed in Subsec.~\ref{subsec:sg-locality-main}: 
reflection terms are additional local degrees of freedom available to the code construction.

\section{Numerical results}
\label{sec:numerical-results}

\subsection{Performance of topological space-group codes during size scaling}

We now present numerical results in Fig.~\ref{fig:size-scaling} comparing 
reflection-code families that pass the algebraic topological-order test with those that fail it.
Specifically, we consider three two-dimensional topological reflection codes and three non-topological reflection codes.
Figure~\ref{fig:size-scaling}(a) shows the finite-size logical-qubit number for the topological codes.
As $L$ varies, $k_L$ oscillates,
sometimes reaching the corresponding infinite-size upper bound and never exceeding it, 
in agreement with the infinite-size cokernel-sector calculation.
The plotted values are obtained from the direct rank computation of the finite binary check matrices.
As an independent check, we also applied the finite-size module Gr\"{o}bner method of 
Theorem~\ref{thm:finite-periodic-logical-count}; for the tested square sizes, 
it gives exactly the same $k_L$ as the direct binary-rank calculation.
Figure~\ref{fig:size-scaling}(b) shows the corresponding code distances.
The distance grows with $L$ with an oscillatory component superposed 
on an approximately linear trend over the tested square sizes.
The observed distance behavior is consistent with topological codes: local periodicity 
can change the shortest logical representatives at special sizes, 
while the overall scale of the logical distance increases with the linear system size in these examples.

Figure~\ref{fig:size-scaling}(c)(d) shows the contrasting behavior of the non-topological codes.
Their logical-qubit numbers grow with the linear system size, again with parity-dependent oscillations, rather than staying below a fixed sector-counting bound.
At the same time, their distances remain small over the tested range.

\begin{figure}[t]
\centering
\includegraphics[width=\columnwidth]{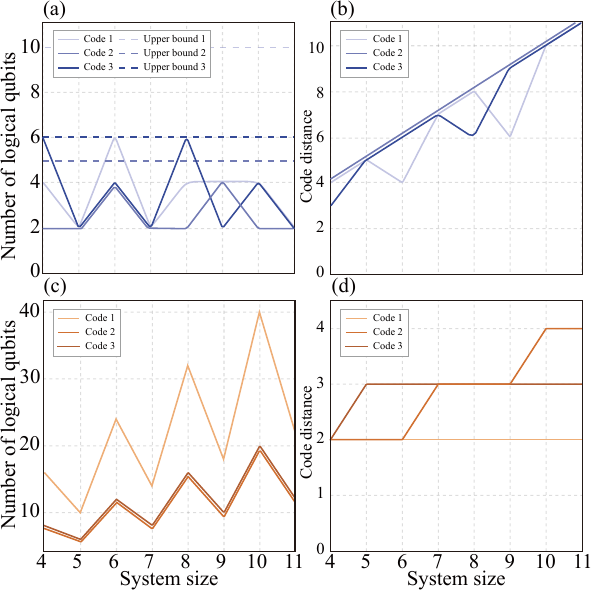}
\caption{
Numerical scaling of code properties with system size $L$ (including $2L^2$ physical qubits)
for (a)--(b) topological reflection codes and (c)--(d) non-topological reflection codes.
(a),(c) The number of logical qubits with respect to $L$;
The dashed horizontal lines in (a) indicate the upper bounds obtained from 
the infinite-size cokernel-sector calculation based on Eq.~(\ref{eq:sector-number}).
(b),(d) The code distance as a function of $L$.
For the topological codes in (a) and (b), the three data sets correspond to the following choices of 
$(f,h)$ from Eq.~(\ref{eq:fxfzhxhz}):
$(1+t_y, 1+t_xt_y^2+s_xt_y+t_x^3)$,
$(1+t_y, 1+t_x+s_xt_y+s_xt_x^2t_y^2)$, and
$(1+t_y, 1+t_x+s_xt_y+s_xt_x^3t_y)$.
For the non-topological codes in (c) and (d), the three choices are
$(t_y+t_y^{-1},t_xt_y^{-1}+t_x^{-2}+s_xt_x^{-1}t_y+s_xt_x^2t_y^2)$,
$(t_y+t_y^{-1},s_xs_yt_x^{-2}+s_xs_yt_x^{-2}t_y+s_yt_y^{-1}+s_yt_y^{-2})$, and
$(t_y+t_y^{-1},s_xs_yt_xt_y^2+s_xs_yt_x^{-1}t_y+s_yt_x+s_yt_x^{-1}t_y^{-1})$.
}
\label{fig:size-scaling}
\end{figure}

\subsection{Codes with outstanding code parameters}

Extending BB codes to space-group codes substantially enlarges the accessible code family.
The most immediate benefit is the appearance of new codes whose parameters outperform those of BB codes with the same number $n$ of physical qubits.
Following the usual convention, we use $\eta=\frac{k d^2}{n}$ as the figure of merit for the code parameters.
We perform numerical searches for space-group codes with stabilizer weights 6 and 8, with representative results summarized in
Table~\ref{tab:twenty-standard-sg-codes-matched-bb}.
All codes in the table use
\begin{equation} \label{eq:fxfzhxhz}
f_X=f^T,\qquad h_X=h^T,\qquad f_Z=h,\qquad h_Z=f.
\end{equation}
Boldface identifies strict improvements over the explicitly reported same-$n$, same-weight BB benchmarks in Refs.~\cite{liang2025generalized,liang2025self}.
Every code distance reported in the table was certified exactly using Boolean satisfiability problem (SAT) solving.

\begin{table*}[htbp]
\caption{
Our certified space-group codes of stabilizer weight 6 or 8, sorted by blocklength $n$.
The columns $f$ and $h$ are the two support polynomials defining $H_X$ by $(f^T,h^T)$ and $H_Z$ by $(h,f)$.
The rectangular point group $C_{2v}=\{e,s_x,s_y,s_xs_y\}$ consists of the identity, 
two perpendicular mirror reflections $s_x,s_y$, and their product, a $180^\circ$ rotation.
For the $n=96$ row, the reflection operation $s$ transforms $(x,y)$ to $(x,-x-y)\pmod{(24,2)}$.
Every displayed distance was determined exactly by Boolean satisfiability problem (SAT) solving. 
A bold code has strictly larger integer $kd^2$ than the explicitly reported same-$n$, 
same-stabilizer-weight row in the specified primary baseline~\cite{liang2025generalized,liang2025self}.}
\label{tab:twenty-standard-sg-codes-matched-bb}
\scriptsize
\setlength{\tabcolsep}{2pt}
\renewcommand{\arraystretch}{1.3}
\begin{ruledtabular}
\begin{tabular}{ccccc}
$[[n,k,d]]$ & group & $f$ & $h$ & $kd^2/n$\\
\hline
$[[12,8,2]]$ & $(\mathbb Z_{2}\times\mathbb Z_{3})\rtimes C_s$ & $s_yt_y+s_yt_y^2+t_xs_yt_y+t_xs_yt_y^2$ & $1+t_y^2+t_x+t_xt_y^2$ & $2.7$\\
$[[14,6,2]]$ & $\mathbb Z_{7}\rtimes C_s$ & $s_x+s_x t_x^{2}+s_x t_x^{6}$ & $s_x t_x^{3}+s_x t_x^{4}+s_x t_x^{6}$ & $1.7$\\
$[[14,8,2]]$ & $\mathbb Z_{7}\rtimes C_s$ & $s_x+s_x t_x+s_x t_x^2+s_x t_x^4$ & $s_x+s_x t_x+s_x t_x^2+s_x t_x^4$ & $2.3$\\
$[[18,4,4]]$ & $(\mathbb Z_{3}\times\mathbb Z_{3})\rtimes C_s$ & $s_xt_x+s_xt_x^{2}+s_xt_x^{2}t_y$ & $s_xt_y+s_xt_y^{2}+s_xt_xt_y^{2}$ & $3.6$\\
$[[18,6,3]]$ & $(\mathbb Z_{3}\times\mathbb Z_{3})\rtimes C_{2v}$ & $s_y+t_xs_y+t_xs_yt_y+t_x^2s_yt_y$ & $s_x+s_xt_y^2+s_xt_x+s_xt_xt_y$ & $3.0$\\
$[[24,4,4]]$ & $(\mathbb Z_{3}\times\mathbb Z_{4})\rtimes C_{2v}$ & $s_yt_y+s_x+s_xt_x^{2}s_y$ & $s_yt_y^{3}+s_x+s_xt_x^{2}s_y$ & $2.7$\\
$[[\textbf{24,10,4}]]$ & $(\mathbb Z_{6}\times\mathbb Z_{2})\rtimes C_s$ & $s_x+s_xt_y+s_xt_x^2+s_xt_x^5t_y$ & $s_xt_x+s_xt_xt_y+s_xt_x^2t_y+s_xt_x^5t_y$ & $6.7$\\
$[[28,6,4]]$ & $(\mathbb Z_{7}\times\mathbb Z_{2})\rtimes C_s$ & $s_xt_y+s_xt_x^{2}+s_xt_x^{3}t_y$ & $s_xt_xt_y+s_xt_x^{3}+s_xt_x^{4}$ & $3.4$\\
$[[\textbf{30,10,4}]]$ & $(\mathbb Z_{5}\times\mathbb Z_{3})\rtimes C_s$ & $s_yt_y+s_yt_y^2+t_xs_yt_y+t_x^2s_y$ & $t_xs_y+t_x^2s_y+t_x^2s_yt_y^2+t_x^3s_y$ & $5.3$\\
$[[\textbf{32,6,6}]]$ & $\mathbb Z_{16}\rtimes C_s$ & $t_x^{7} s_x+t_x^{8} s_x+t_x^{9}s_x+t_x^{14}s_x$ & $t_x s_x+t_x^{11}s_x+t_x^{12}s_x+t_x^{14}s_x$ & $6.8$\\
$[[36,4,6]]$ & $(\mathbb Z_{3}\times\mathbb Z_{6})\rtimes C_{2v}$ & $t_xs_y+s_xs_yt_y+s_xt_xs_yt_y$ & $t_x+s_xt_y^{5}+s_xt_xt_y$ & $4.0$\\
$[[\textbf{40,14,4}]]$ & $(\mathbb Z_{4}\times\mathbb Z_{5})\rtimes C_{2v}$ & $t_x^2t_y^3+t_y^2s_x+t_xs_xs_y+t_xt_ys_xs_y$ & $t_x^2t_y^3+t_y^2s_x+t_x^3s_xs_y+t_x^3t_ys_xs_y$ & $5.6$\\
$[[42,6,6]]$ & $(\mathbb Z_{7}\times\mathbb Z_{3})\rtimes C_s$ & $t_y+t_x^{4}s_yt_y^{2}+t_x^{5}$ & $t_x^{2}s_yt_y^{2}+t_x^{3}+t_x^{5}t_y^{2}$ & $5.1$\\
$[[48,4,8]]$ & $(\mathbb Z_{6}\times\mathbb Z_{4})\rtimes C_{2v}$ & $s_xt_y+s_xt_xs_yt_y^{2}+s_xt_x^{2}$ & $s_xt_y+s_xt_x^{2}t_y^{3}+s_xt_x^{4}$ & $5.3$\\
$[[\textbf{50,18,4}]]$ & $(\mathbb Z_{5}\times\mathbb Z_{5})\rtimes C_s$ & $s_xt_y^3+s_xt_y^4+s_xt_xt_y^2+s_xt_x^4$ & $s_x+s_xt_y^3+s_xt_x^2t_y+s_xt_x^3t_y^2$ & $5.8$\\
$[[60,8,6]]$ & $(\mathbb Z_{6}\times\mathbb Z_{5})\rtimes C_s$ & $1+t_y^{3}+t_xs_yt_y^{4}$ & $t_y^{2}+t_xs_yt_y^{4}+t_x^{4}$ & $4.8$\\
$[[70,6,8]]$ & $(\mathbb Z_{5}\times\mathbb Z_{7})\rtimes C_s$ & $t_x^{3}t_y^{6}+s_xt_xt_y+s_xt_x^{2}t_y^{2}$ & $t_x^{2}t_y^{6}+s_xt_xt_y+s_xt_x^{2}t_y^{2}$ & $5.5$\\
$[[\textbf{84,10,9}]]$ & $(\mathbb Z_{14}\times\mathbb Z_{3})\rtimes C_s$ & $t_x^6t_y^2+t_x^7+t_x^8t_y+t_xt_y^2s_y$ & $t_x+t_x^{12}t_y+t_x^{12}t_y^2+t_x^{13}$ & $9.6$\\
$[[\textbf{96,8,10}]]$ & $(\mathbb Z_{24}\times\mathbb Z_{2})\rtimes C_s$ & $1+t_x+s t_x^{17}t_y$ & $1+t_x^{2}t_y+t_x^{22}$ & $8.3$\\
$[[112,6,12]]$ & $(\mathbb Z_7\times\mathbb Z_4)\rtimes C_s$ & $1+t_x^6s_y+t_x^4t_y^2s_y$ & $1+t_x^4t_y^3+t_x^6t_y^3 s_y$ & $7.7$\\
$[[\textbf{112,16,10}]]$ & $(\mathbb Z_7\times\mathbb Z_4)\rtimes C_s$ & $1+t_x^2t_y s_y+t_x^4 s_y+t_x^3t_y^2 s_y$ & $1+t_xt_y^2 s_y+t_x^4t_y+t_x^6t_y^3$ & $14.3$\\
$[[\textbf{120,14,12}]]$ & $(\mathbb Z_{15}\times\mathbb Z_4)\rtimes C_s$ & $s_y+t_x^7t_y^2+t_x^{10}t_y^3 s_y+t_x^{12}$ & $t_x^2+t_x^7+t_x^8t_y^3 s_y+t_x^{13}t_y^2 s_y$ & $16.8$\\
$[[\textbf{132,8,11}]]$ & $(\mathbb Z_{22}\times\mathbb Z_3)\rtimes C_s$ & $1+t_x^{13} s_y+t_x^{13}t_y s_y+t_x^{16}t_y$ & $1+t_x^{20}$ & $7.3$\\
$[[150,8,12]]$ & $(\mathbb Z_{15}\times\mathbb Z_5)\rtimes C_s$ & $t_y^3+t_x^3+t_x^9t_y^2 s_y$ & $t_x^5t_y^2+t_x^6+t_x^{14}t_y^2 s_y$ & $7.7$\\
$[[\textbf{156,8,12}]]$ & $(\mathbb Z_{26}\times\mathbb Z_3)\rtimes C_s$ & $1+t_x^{24}t_y^2+t_x^9 s_y+t_x^3 t_y^2 s_y$ & $1+t_x^2$ & $7.4$\\
$[[\textbf{160,26,8}]]$ & $(\mathbb Z_{20}\times\mathbb Z_4)\rtimes C_s$ & $t_xt_y^3+t_x^{16}s_y+t_x^{14}+t_x^{19}t_y^2$ & $t_xs_y+t_x^3+t_x^6t_y+t_x^8$ & $10.4$\\
\end{tabular}
\end{ruledtabular}
\end{table*}

\subsection{Advantage in AOD movement}
\label{sec:folded-aod-movement-results}

In this subsection, we demonstrate how the abstract locality advantage of space-group codes can be converted into a scheduling advantage for syndrome extraction on reconfigurable neutral-atom arrays. 
Specifically, we numerically compare the AOD movement required by parameter-matched BB and reflection codes under the folded two-AOD model introduced in Sec.~\ref{sec:movement-illustration}. 
The paired codes have identical finite-code parameters $[[n,k,d]]$. 
Across several system sizes, the reflection codes admit syndrome-extraction schedules with shorter cumulative movement
than the optimized BB comparators.

We assume space-group codes and BB codes have normal form as Eq.~(\ref{eq:fxfzhxhz}).
Writing
\begin{equation}
f=f_1+f_2+f_3,
\qquad
h=h_1+h_2+h_3,
\end{equation}
each check-ancilla qubit visits six data qubits in each CSS sector. 
A representative cyclic ordering is
\begin{equation}
f_1\longrightarrow f_2\longrightarrow f_3
\longrightarrow h_1\longrightarrow h_2\longrightarrow h_3
\longrightarrow f_1.
\label{eq:six-term-cyclic-aod-schedule}
\end{equation}
The final arrow is the round-closure transition and is included in the movement cost. 
For each code, we consider all possible cyclic ordering and find the one
with minimum movement cost.

Specifically, for a transition from a group element $u$ to a 
group element $v$, one needs to realize the relative operation $vu^{-1}$ 
by moving check arrays. 
The arrays can be displaced in parallel along the $x$ direction or the $y$ direction, 
while avoiding any crossings during the motion.
The cost of such a parallel movement is determined by the longest distance 
traveled by any single array. We enumerate all valid movement sequences 
and selects the one that minimizes this longest distance; the resulting minimum cost is denoted by
$W(vu^{-1})$, expressed in line-slot displacement units.
A full syndrome-extraction cycle corresponds to a cyclic ordering of the six elements.
Consider an order
\begin{equation}
	v^{X}_1\longrightarrow v^{X}_2 \longrightarrow v^{X}_3
	\longrightarrow v^{X}_4 \longrightarrow v^{X}_5 \longrightarrow v^{X}_6
	\longrightarrow v^{X}_1,
\end{equation}
where $v^{X}_t$ with $t=1,2,\dots,6$ belong to the set $\{f_1,f_2,f_3,h_1,h_2,h_3\}$,
and similarly for $v^{Z}_t$. 
For a given cyclic order, we define the cumulative array movement over a syndrome-extraction cycle as
\begin{equation}
L_{\mathrm{AOD}}=\sum_{\sigma\in\{X,Z\}}\sum_{t=1}^{6}W(v^{\sigma}_{t+1}(v^{\sigma}_{t})^{-1}),
\label{eq:folded-two-aod-movement}
\end{equation}
where $v^{\sigma}_7 \equiv v^{\sigma}_1$.
The minimum movement cost $L_{\mathrm{m}}$ is then obtained by minimizing $L_{\mathrm{AOD}}$
over all possible cyclic orderings.

Figure~\ref{fig:folded-aod-movement-comparison} shows the reductions for the reflection codes 
at $n=48,56,96,108,$ and $132$ compared with the corresponding BB codes. 
The largest reduction occurs at $n=108$: the reflection code 
requires $L_{\mathrm{m}}=80$, whereas the optimized BB comparator requires 
$L_{\mathrm{m}}=200$. 
The corresponding reduction places the reflection-code movement below one half of the BB value. 
The group-algebra polynomials for all five pairs are listed in Table~\ref{tab:folded-aod-code-pairs},
and the BB codes are taken from Ref.~\cite{wang2026coprime}.

\begin{figure}[t]
\centering
\includegraphics[
width=0.45\textwidth
]{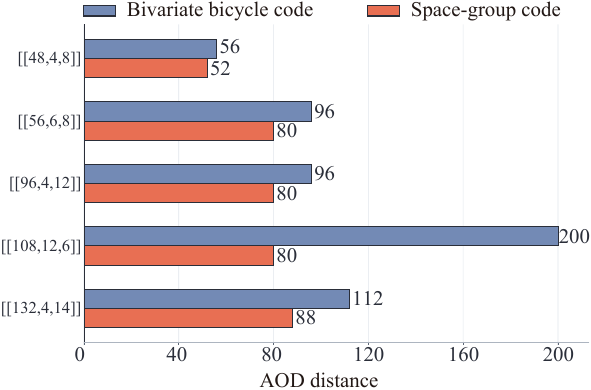}
\caption{
Folded two-AOD movement for five parameter-matched reflection-code and BB-code pairs.
Each pair has identical exact parameters $[[n,k,d]]$. 
The horizontal bars represent the minimum movement cost $L_{\mathrm{m}}$ 
during a syndrome-extraction cycle. 
}
\label{fig:folded-aod-movement-comparison}
\end{figure}

\begin{table}[t]
\caption{
Group-algebra representatives for the five matched code pairs containing one reflection code and one BB code, which appear in Fig.~\ref{fig:folded-aod-movement-comparison}.
}
\label{tab:folded-aod-code-pairs}
\scriptsize
\setlength{\tabcolsep}{2pt}
\renewcommand{\arraystretch}{1.3}
\begin{tabular}{cccc}
\hline
Code & Shape & $f$ & $h$\\
\hline
$[[48,4,8]]$ & $3\times8$ & $1+t_xt_y^6+t_x^2t_y^6$ & $1+t_x^2t_y^6+s_xt_y^7$ \\
$[[48,4,8]]$ & $3\times8$ & $1+t_xt_y+t_x^2t_y^2$ & $1+t_xt_y^2+t_x^2t_y^2$ \\
\hline
$[[56,6,8]]$ & $4\times7$ & $s_xt_y^4+s_xt_x^3t_y+s_xt_x^3t_y^3$ & $t_y+t_x^2t_y^3+t_x^2t_y^4$\\
$[[56,6,8]]$ & $4\times7$ & $t_y+t_x+t_x^3t_y^3$ & $1+t_xt_y+t_x^3t_y^3$ \\
\hline
$[[96,4,12]]$ & $6\times8$ & $1+t_x^2t_y+t_x^4t_y$ & $1+t_x^2t_y^2+s_xt_xt_y$\\
$[[96,4,12]]$ & $6\times8$ & $t_y^2+t_xt_y^7+t_x^5t_y^2$ & $1+t_x^4t_y+t_x^5t_y^2$\\
\hline
$[[108,12,6]]$ & $6\times9$ & $1+t_y+s_xt_x^3t_y^5$ & $t_y+s_xt_x^3t_y^2+s_xt_x^3t_y^3$\\
$[[108,12,6]]$ & $2\times27$ & $1+t_y^{24}+t_xt_y^{21}$ & $1+t_xt_y^{15}+t_xt_y^{21}$\\
\hline
$[[132,4,14]]$ & $6\times11$ & $1+t_x^4t_y^3+s_xt_x^5$ & $t_y+t_x^4t_y^3+s_xt_x^5$ \\
$[[132,4,14]]$ & $6\times11$ & $1+t_x+t_x^2t_y^{10}$ & $1+t_xt_y^2+t_x^2t_y^4$\\
\hline
\end{tabular}
\end{table}

A lower value of $L_{\mathrm{m}}$ represents a smaller cumulative transport burden for the movable check arrays. 
Under fixed line spacing, kinematic constraints, and serialization rules, this reduction can translate into a shorter movement contribution to the syndrome-extraction cycle and less exposure to movement-induced loss and decoherence. 
Reducing the AOD movement distance implies a shorter time for syndrome extraction, which results in a lower physical error rate, thereby extending the quantum coherence time.
More broadly, an expanded numerical search may identify further space-group codes that complement BB codes across movement, code
quality, and hardware-constrained optimization objectives.

\subsection{Advantage in superconducting layout}
\label{sec:sg-superconducting-layouts}

In this subsection, we show how the abstract locality advantage of space-group codes can be converted into a concrete layout advantage on superconducting hardware~\cite{divincenzo2009fault,zhao2022realization,marques2022logical,krinner2022realizing,acharya2023suppressing,google2025quantum,wang2026demonstration}.  
In particular, we find numerically that a lossless folded-lattice placement of the Gross code substantially reduces
its routing cost relative to the layout reported in Ref.~\cite{mathews2026placing}.  
We also construct and evaluate superconducting layouts for two additional reflection codes, $[[48,4,8]]$ and $[[96,6,10]]$.

We first summarize the hardware model used to quantify the layout overhead of a stabilizer code.  
The superconducting platform consists of a stack of chips, each with two routable faces.  
Two mutually facing chip surfaces form a tier.  
Transmons and electromagnetic coupling traces can be placed on a chip surface, with each transmon representing one physical qubit.  
We place all transmons on the upper face of tier~0 and partition them into data transmons and check (ancilla) transmons.  Syndrome extraction requires a coupling path between every check transmon and each data transmon in the support of the corresponding stabilizer.  
Traces on the same chip face are required to be noncrossing.  
To resolve otherwise unavoidable crossings, a route may switch between the two faces of a tier through a bump bond, or
pass between the two faces of the same chip through a through-silicon via (TSV).

Following Ref.~\cite{mathews2026placing}, we characterize a routed layout by four extracted hardware parameters.  
The tier count includes the qubit tier~0.
The quantity length is the normalized average coupler length obtained by 
the Hardware-Aware Layout (HAL) algorithm~\cite{mathews2026placing}, using the shortest allowed tier 0 route as the length unit.
The quantity bumps is the maximum, over all tiers, of the tier-averaged HAL face-switch (bump-transition) score.  
Finally, $V$ is the average number of TSV transitions per edge assigned above tier~0.  
These definitions and the two-decimal display precision are the same as those used in the extracted-hardware-parameter
table of Ref.~\cite{mathews2026placing}.

The HAL algorithm takes a Tanner graph and a set of initial qubit positions as input and 
uses heuristic placement and routing to produce a feasible set of coupler paths across the available tiers~\cite{mathews2026placing}.  
We retain their routing algorithm and hardware settings, but replace the initial positions by our folded-lattice placements
with data and check transmons occupying distinct sites.  
The periodic translation cells are embedded in an open plane using a two-axis accordion fold.  

Figures~\ref{fig:sg-hal-layouts}(k)--(o) show the resulting five-tier routing of the Gross $[[144,12,12]]$ code.  
Using the definitions and display precision of Ref.~\cite{mathews2026placing}, the published Gross-code layout has
$(\text{Tiers},\text{Length},\text{Bumps},\text{TSVs})=(5,11.08,5.06,3.27)$, whereas our folded placement gives $(5,7.09,5.14,3.06)$.  
Thus, the number of tiers is unchanged, the normalized coupler length is reduced by approximately $36.0\%$, and the TSV metric
decreases from $3.27$ to $3.06$, while the maximum tier-averaged face-switch score changes only slightly.

We perform the same calculation for the two reflection codes.
Figures~\ref{fig:sg-hal-layouts}(a)--(e) show the HAL routes for $[[48,4,8]]$, while Figs.~\ref{fig:sg-hal-layouts}(f)--(j) show those for $[[96,6,10]]$.  
Their defining supports and extracted routing overheads are listed in Table~\ref{tab:sg-hal-overheads}.  
Both codes admit five-tier layouts.  
These results provide concrete evidence that the reflection structure of space-group codes, together with an orbifold-inspired folded placement, can be exploited at the hardware-layout level.  
They also provide practical reference points for future implementations of high-performance qLDPC codes on multilayer superconducting circuits.

\begin{figure*}[t]
\centering
\includegraphics[width=1\textwidth]{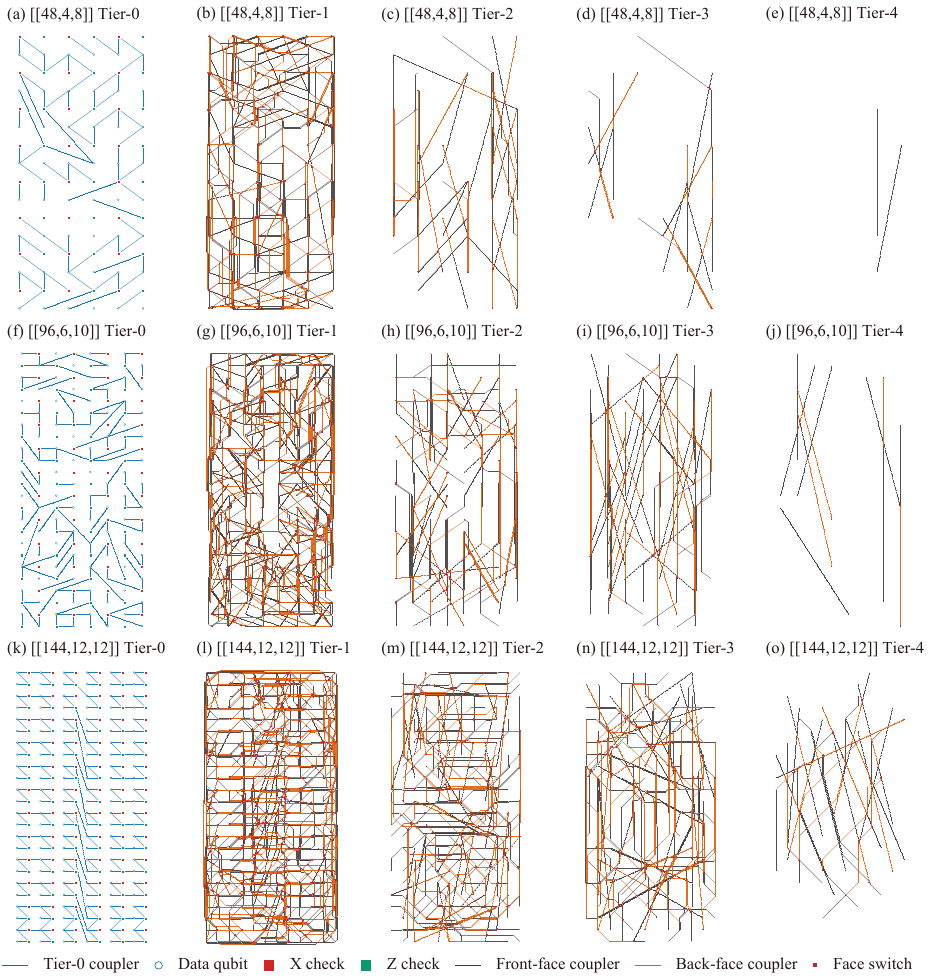}
\caption{
HAL-generated with folded-lattice superconducting layouts for two reflection space-group codes and the Gross code.  
(a)--(e) the $[[48,4,8]]$ space-group code;
(f)--(j) the $[[96,6,10]]$ space-group code; 
(k)--(o) the Gross $[[144,12,12]]$ BB code.
For each code, columns~0--4 show routing tiers~0--4.
In tier~0, blue traces denote couplers, open blue circles denote
data transmons, and the larger red and green squares denote $X$ and $Z$ check transmons, respectively.
In tiers~1--4, black and orange traces denote couplers on the
front and back chip faces, while the small red squares mark face
switches between the two faces.
Each Tanner edge is routed within a single assigned tier; the
vertical TSV access stacks connecting tier~0 to higher routing
tiers are implicit and are not shown.
}
\label{fig:sg-hal-layouts}
\end{figure*}

\begin{table*}[t]
\caption{
Extracted hardware parameters for the folded-lattice layouts, using the definitions and display precision of Ref.~\cite{mathews2026placing}.
Here $\text{Tiers}$ is the total number of tiers including tier~0, $\text{Length}$ is the normalized average coupler length, $\text{Bumps}$ is the maximum tier-averaged HAL face-switch score, and $\text{TSVs}$ is the average number of TSV transitions per edge assigned above tier~0.  
}
\label{tab:sg-hal-overheads}
\renewcommand{\arraystretch}{1.3}
\begin{ruledtabular}
\begin{tabular}{ccccccc}
Code & $f$ & $h$ & Tiers & Length & Bumps & TSVs \\
\hline
$[[48,4,8]]$
& $t_y^3s_xs_y+t_xt_y^3s_xs_y+t_x^2t_ys_xs_y$
&  $t_xt_y^3+t_x^2t_y^5+t_x^2s_x$
& 5 & 4.55 & 3.36 & 2.70 \\
$[[96,6,10]]$
& $t_x^4t_ys_y+t_x^5t_y^3$ 
& $t_ys_y+t_xt_y+t_x^3+t_x^7t_y^2$
& 5 & 6.95 & 3.54 & 2.93 \\
$[[144,12,12]]$
& $t_y^3+t_x+t_x^2$ 
& $t_x^3+t_y+t_y^2$
& 5 & 7.09 & 5.14 & 3.06 \\
\end{tabular}
\end{ruledtabular}
\end{table*}

\section{Discussion}
\label{sec:discussion}

The principal result of this work is the construction and characterization of topological codes beyond ordinary translation invariance.
We show that nontrivial crystallographic point-group operations can enter the bulk stabilizers while preserving both topological order and geometric locality. 
The enlarged search space yields space-group codes that improve upon reported same-number-of-physical-qubits and same-weight BB benchmarks. 
The hardware studies further show that this additional structure can have operational consequences: reflection codes reduce optimized movement costs in neutral-atom arrays, and folded geometries enable practical multilayer superconducting layouts.
Translation symmetry has traditionally played two roles at once: it organizes the stabilizer algebra and provides an immediate notion of locality.
Our framework separates these roles. 
Topological order can be characterized algebraically even when ordinary translation invariance is broken, while locality can instead be expressed through an orbifold geometry. 
Spatial symmetry therefore becomes a design resource. 
The improved folded placement of the Gross BB code further shows that this geometric viewpoint can also provide new implementations of established translation-invariant codes. 
The broader contribution of the present work is thus a new symmetry-based direction for the co-design of topological codes and quantum hardware.

An important next step is to determine more advantages in end-to-end fault-tolerant performance. 
This will require circuit-level noise simulations, hardware-calibrated syndrome-extraction models, and decoders adapted to the reduced translation symmetry of space-group codes~\cite{yin2024symbreak,gu2026scalable}. 
The additional spatial symmetries may also support useful logical operations through qubit relabeling or code automorphisms~\cite{koh2026entangling}. 
On the theoretical side, it will be valuable to extend the construction to nonsymmorphic space groups~\cite{shiozaki2016topology,liu2014topological}, higher-dimensional~\cite{hastings2016quantum,bombin2007exact,hamma2005string,dennis2002topological}, defects-containing~\cite{webster2020fault,brown2020parallelized} and layered systems~\cite{liu2026self}, and multiblock CSS constructions~\cite{lin2026abelian}.

\textit{Note added.} 
In the final stage of completing this work, we became aware of Ref.~\cite{aydin2026breakingbicycleframecosetbased}, which also investigates the construction of qLDPC codes based on groups and their cosets. 
Starting from a related code construction, our work focuses mainly on topological properties of these codes and their compatibility with hardware constraints, whereas Ref.~\cite{aydin2026breakingbicycleframecosetbased} focuses primarily on the search for codes with good parameters, as well as decoding and simulations under a circuit-level error model.

\begin{acknowledgments}
We thank P. Zoller, Y. Nakamura, and Z. Yue for helpful discussions. This work is supported by 
Quantum Science and Technology-National Science and Technology Major Project (Grant No. 2021ZD0301604)
and the National Natural Science Foundation of China (Grant No. 11974201).
We also acknowledge the support by center of high performance computing, Tsinghua University.
\end{acknowledgments}

\appendix

\section{Coset space of a group used for code construction}
\label{app:generate-set}

In the main text, we construct CSS codes from a set $A$ equipped with a permutation action of 
a group $G$.
In this Appendix, we show how this construction can be achieved directly from  
coset spaces of $G$.

It is well known that every $G$-set $A$ can be decomposed into a disjoint union of orbits $O_i$~\cite{gallian2021contemporary}:
\begin{equation}
	A=\bigsqcup_i O_i,\qquad O_i=G a_i=\{g a_i \mid g\in G\},
	\label{eq:g-set-orbit-decomposition}
\end{equation}
where acting each group element on an orbit element sends this orbit 
element to an element in the same orbit.

For any subgroup $H_i$ in $G$, its conjugacy class is a set $S_{H_i}=\{g H_i g^{-1}:g\in G\}$. 
Let $H_i$ be a representative of each conjugacy class of subgroups.
We want to show that every $G$-set $A$ admits a decomposition of the form
\begin{equation}
	A\cong \bigsqcup_{i=1}^m \bigsqcup_{j=0}^{n_i-1} [G/H_i]_j,
	\label{eq:gset-general-decomposition}
\end{equation}
where $[G/H_i]_0=G/H_i$ and for $j>0$
$[G/H_i]_j$ denotes the $j$th copy of the coset space $G/H_i$.

As shown in Subsec.~\ref{subsec:operator-group} of the main text, 
a CSS code can be constructed from four group-algebra 
elements $f_X,h_X,f_Z,h_Z\in\F[G]$ with the check matrices 
given in Eq.~(\ref{eq:gamma-block-maps}).
Relative to the disjoint union of coset spaces in Eq.~(\ref{eq:gset-general-decomposition}),
each element $g\in G$ is represented by a permutation matrix $\Gamma(g)$. In addition,
each coset in Eq.~(\ref{eq:gset-general-decomposition}) corresponds to two physical 
qubits.

Thus, the key input is the decomposition in Eq.~(\ref{eq:gset-general-decomposition}),
which itself follows directly from the following two propositions.
Although their proofs can be found in standard textbooks~\cite{Artin2014}, 
we include them here for completeness and self-consistency. 
We first observe that every subgroup $H\subseteq G$ gives rise to a left coset space 
$G/H$, which carries a natural $G$-set structure: left multiplication by any $g \in G$
permutes the cosets. Indeed, for distinct cosets $g_1 H\neq g_2 H$, their images 
$g (g_1H)$ and $g (g_2 H)$ are distinct. 

\begin{proposition}
	\label{thm:orbit-coset}
	Let $A$ be a $G$-set. Let $O=Ga$ be an orbit with $a\in A$ and let $H=\{g\in G\mid ga=a\}$ be the
	corresponding stabilizer subgroup.
	Then $O$ is isomorphic to the coset space $G/H$ with an isomorphism defined by
	\begin{equation}
		\Psi_a:G/H\longrightarrow O,\qquad \Psi_a(gH)=ga.
		\label{eq:coset-orbit-map}
	\end{equation}
\end{proposition}

\begin{proof}
	If $gH=g'H$, then $g^{-1}g'\in H$, hence $(g^{-1}g')a=a$ and $g'a=ga$.
	Thus $\Psi_a$ is well-defined.
	$\Psi_a$ is surjective because every element in $O$ is of the form $g a$.
	If $\Psi_a(gH)=\Psi_a(g'H)$, then $ga=g'a$, so $g^{-1}g'a=a$ and $g^{-1}g'\in H$.
	Hence $gH=g'H$ so that $\Psi_a$ is injective.
	$\Psi_a$ is thus bijective.
	For any $k\in G$ and $gH\in G/H$, one has
	\begin{equation}
		\Psi_a(k\cdot(gH))=\Psi_a((kg)H)=(kg)a=k(ga)=k\cdot \Psi_a(gH).
	\end{equation}
	Hence $\Psi_a$ is $G$-equivariant. Thus, $\Psi_a$ is an isomorphism and
	$O\cong G/H$.
\end{proof}

\begin{proposition}
	\label{thm:conjugate-subgroups-gsets}
	For two subgroups $H,K \le G$, the $G$-sets $G/H$ and $G/K$ are isomorphic if and only if $H$ and $K$ are conjugate.
\end{proposition}

\begin{proof}
	Assume that $\Phi:G/K\to G/H$ is a $G$-set isomorphism and write $\Phi(K)=gH$.
	For any $k\in K$, we have $\Phi(K)=gH=\Phi(k K)=k\Phi(K)=k gH$, yielding that $g^{-1}kg \in H$
	and thus $g^{-1}Kg\subseteq H$.
	Applying the same argument to $\Phi^{-1}$ gives the reverse inclusion, so $g^{-1}Kg=H$.
	Conversely, if $gKg^{-1}=H$, then the map
	\begin{equation}
		h K\longmapsto h g^{-1}H
	\end{equation}
    with $h \in G$ is a well-defined $G$-equivariant bijection from $G/K$ to $G/H$.
\end{proof}

We note that quantum two-block group algebra codes (2BGA) codes~\cite{lin2024quantum}
are constructed from regular representation of a finite group, i.e., choosing a subgroup $H_i=\{e\}$, 
where $e$ is the identity element in the group. In addition, Z semidirect Z (ZSZ) codes~\cite{guo2026toward} 
form a subclass of 2BGA codes. They are obtained by choosing the underlying group to be a 
semidirect product $G=\mathbb Z_\ell\rtimes \mathbb Z_m$.
This gives a non-abelian generalization of the BB construction. 
When the semidirect product action is trivial, the group reduces to the 
direct product $\mathbb Z_\ell\times\mathbb Z_m$.

Multivariate bicycle codes~\cite{voss2025multivariate} are obtained by taking the 
underlying group to be an abelian translation group, for example 
$G=\mathbb Z_{L_1}\times\cdots\times\mathbb Z_{L_D}$ in the finite periodic case. 
Equivalently, the blocks are described by multivariate cyclic polynomials in $D$ commuting variables.
The BB codes are the two-dimensional special case of multivariate bicycle codes.

\section{Basis table for two-dimensional point groups}
\label{app:basis-table}

In this Appendix, we provide the invariant polynomial subrings and free module bases used in the module-lifting step of Subsec.~\ref{subsec:module-lifting}.
For a two-dimensional point group $P$, the Laurent polynomial ring
\begin{equation}
	\mathcal R_0=\F[x^{\pm1},y^{\pm1}]
\end{equation}
is treated as a finite free module over a polynomial subring $\mathcal{R}\subseteq\mathcal R_0^P$.
After a basis $E_P$ is fixed, every translation or point-group generator acts by a finite matrix over $\mathcal{R}$.
This is the input used to build the maps $\Lambda_X$ and $\Lambda_Z$.

\begin{table*}[!htbp]
\caption{
Chosen invariant polynomial subrings and free-module bases for two-dimensional
crystallographic point-group actions.
For each action $P$, the displayed ring satisfies
$\mathcal R\subseteq\mathcal R_0^P$, and
$\mathcal R_0=\bigoplus_{e\in E_P}\mathcal R e$.
The superscripts on $C_s$ and $C_{2v}$ distinguish their integral-lattice embeddings,
as defined in Eqs.~(\ref{eq:point-group-basis-actions}) and
(\ref{eq:schoenflies-point-groups}).
We use
$
	a=x+x^{-1}$,
	$b=y+y^{-1}$,
	$A_2=a+b$,
	$B_2=ab$,
	$s_d=x+y$,
	$p_d=xy$,
	$q_d=p_d+p_d^{-1}$,
	$s_3=x+y+x^{-1}y^{-1}$,
	$p_3=xy+x^{-1}+y^{-1}$,
	$s_6=x+x^{-1}y+y^{-1}$,
	$p_6=y+x^{-1}+xy^{-1}$,
	$A_6=s_6+p_6$, and
	$B_6=s_6p_6$.
For example, for $C_{2v}^{(\mathrm{rect})}$ generated by
$x\mapsto x^{-1}$ and $y\mapsto y^{-1}$, one may take
$\mathcal{R}=\F[a,b]$.
Then the relations $x^2+ax+1=0$ and $y^2+by+1=0$ reduce every Laurent monomial to the basis $\{1,x,y,xy\}$.
}
\label{tab:point-group-basis}
\scriptsize
\renewcommand{\arraystretch}{1.3}
\begin{ruledtabular}
\begin{tabular}{cccc}
point-group action $P$ & $\mathcal{R}$ & basis $E_P$ &
$\operatorname{rank}_{\mathcal R}\mathcal R_0$\\
\hline
$C_s^{(x)}$ & $\F[a,b]$ & $\{1,x,y,xy\}$ & $4$\\
$C_s^{(d)}$ & $\F[s_d,q_d]$ & $\{1,x,xy,x^2y\}$ & $4$\\
$C_2$ & $\F[a,b]$ & $\{1,x,y,xy\}$ & $4$\\
$C_{2v}^{(\mathrm{rect})}$ & $\F[a,b]$ & $\{1,x,y,xy\}$ & $4$\\
$C_{2v}^{(\mathrm{diag})}$ & $\F[A_2,B_2]$ & $\{1,x,y,xy,a,ax,ay,axy\}$ & $8$\\
$C_3$ & $\F[s_3,p_3]$ & $\{1,x,x^2,y,xy,x^2y\}$ & $6$\\
$C_{3v}$ & $\F[s_3,p_3]$ & $\{1,x,x^2,y,xy,x^2y\}$ & $6$\\
$C_4$ & $\F[A_2,B_2]$ & $\{1,x,y,xy,a,ax,ay,axy\}$ & $8$\\
$C_{4v}$ & $\F[A_2,B_2]$ & $\{1,x,y,xy,a,ax,ay,axy\}$ & $8$\\
$C_6$ & $\F[A_6,B_6]$ &
$\{1,x,x^2,y,xy,x^2y,s_6,s_6x,s_6x^2,s_6y,s_6xy,s_6x^2y\}$ & $12$\\
$C_{6v}$ & $\F[A_6,B_6]$ &
$\{1,x,x^2,y,xy,x^2y,s_6,s_6x,s_6x^2,s_6y,s_6xy,s_6x^2y\}$ & $12$\\
\end{tabular}
\end{ruledtabular}
\end{table*}

Table~\ref{tab:point-group-basis} gives one such choice of $\mathcal{R}$ and $E_P$ for each two-dimensional point-group action used in the paper.
The entries are not meant to classify invariant rings uniquely; they provide computationally convenient polynomial subrings over which $\mathcal R_0$ is free.
We use Schoenflies notation: $C_1$ is the trivial point group, $C_n=\langle c_n\rangle$
is generated by an $n$-fold rotation, $C_s$ is generated by a single mirror, and
$C_{nv}$ contains $C_n$ together with $n$ mirror reflections.
Note that $C_s$ and $C_2$ are isomorphic as abstract order-two groups, but the former is
generated by a reflection and the latter by a $180^\circ$ rotation.
The superscripts below distinguish inequivalent integral-lattice embeddings of the same
Schoenflies type; they are embedding labels rather than additional point-group types.
On exponent vectors $(x,y)$, we choose
\begin{align}
s_x(x,y)&=(-x,y),&
s_y(x,y)&=(x,-y),\nonumber\\
s_d(x,y)&=(y,x),&
s_{\bar d}(x,y)&=(-y,-x),\nonumber\\
c_2(x,y)&=(-x,-y),&
c_3(x,y)&=(-y,x-y),\nonumber\\
c_4(x,y)&=(-y,x),&
c_6(x,y)&=(-y,x+y).
\label{eq:point-group-basis-actions}
\end{align}
The groups appearing in the table are
\begin{align}
C_s^{(x)}&=\langle s_x\rangle,&
C_s^{(d)}&=\langle s_d\rangle,\nonumber\\
C_{2v}^{(\mathrm{rect})}&=\langle s_x,s_y\rangle,&
C_{2v}^{(\mathrm{diag})}&=\langle s_d,s_{\bar d}\rangle,\nonumber\\
C_{3v}&=\langle c_3,s_d\rangle,&
C_{4v}&=\langle c_4,s_d\rangle,\nonumber\\
C_{6v}&=\langle c_6,s_d\rangle .
\label{eq:schoenflies-point-groups}
\end{align}
In particular,
$C_{2v}^{(\mathrm{rect})}=\{e,s_x,s_y,c_2\}$, whereas
$C_{2v}^{(\mathrm{diag})}=\{e,s_d,s_{\bar d},c_2\}$.

\section{Gröbner basis for quotient rings and modules}
\label{app:groebner-quotients}

In this Appendix, we review the standard Gröbner-basis algorithm for computing the dimension of a polynomial quotient ring, and then state the corresponding module version~\cite{cox1997ideals,greuel2008singular}.
Let
\begin{equation}
	R=\F[z_1,\ldots,z_m]
\end{equation}
and let $I\subset R$ be an ideal.
Choose a monomial order, namely a total well-order on monomials that is compatible with multiplication.
For a nonzero polynomial
\begin{equation}
	f=\sum_{\alpha} c_{\alpha}z^\alpha ,
\end{equation}
the leading monomial $\operatorname{LM}(f)$ is the largest monomial $z^\beta$ with $c_\beta\ne0$, the leading coefficient is $\operatorname{LC}(f)=c_\beta$, and the leading term is $\operatorname{LT}(f)=c_\beta z^\beta$.
For monomials $z^\alpha$ and $z^\beta$, we say that $z^\alpha$ divides $z^\beta$ if $\alpha_i\le \beta_i$ for every $i$.
In that case the quotient monomial is $z^{\beta-\alpha}$.

For two nonzero polynomials $f,g$, their $S$-polynomial is defined by
\begin{equation}
	S(f,g)=
	\frac{\operatorname{lcm}(\operatorname{LM}(f),\operatorname{LM}(g))}
	{\operatorname{LT}(f)}f
	-
	\frac{\operatorname{lcm}(\operatorname{LM}(f),\operatorname{LM}(g))}
	{\operatorname{LT}(g)}g .
\end{equation}
The purpose of this expression is to cancel the common leading monomial at the least common multiple of the two leading monomials.
Over $\mathbb F_2$, the subtraction is the same as addition.

We also fix the division procedure used below.
To divide a polynomial $h$ by a current list $G=\{g_1,\ldots,g_t\}$, look at the leading term of the current value of $h$.
If some $\operatorname{LM}(g_i)$ divides $\operatorname{LM}(h)$, replace
\begin{equation}
	h
	\longleftarrow
	h-\frac{\operatorname{LT}(h)}{\operatorname{LT}(g_i)}g_i .
\end{equation}
This cancels the current leading term of $h$.
Repeat the same step with the new value of $h$.
If no leading monomial from $G$ divides $\operatorname{LM}(h)$, move $\operatorname{LT}(h)$ to the remainder and remove it from $h$.
The division ends when $h=0$.
Saying that a polynomial reduces to zero means that this division process leaves zero remainder; equivalently, its leading term can be cancelled repeatedly until no term remains outside the span generated by the leading terms of $G$.

Buchberger's algorithm is a loop.
Start with a generating list $G=\{f_1,\ldots,f_s\}$.
For each pair $g_i,g_j\in G$, compute $S(g_i,g_j)$ and divide it by the current list $G$.
If the remainder is nonzero, adjoin this remainder to $G$ and then repeat the same test for the new pairs.
The loop stops when every $S$-polynomial has zero remainder after division by the current list.
By Buchberger's criterion, this final list is a Gröbner basis of $I$~\cite{cox1997ideals,greuel2008singular}.
The process terminates because $R$ is Noetherian~\cite{eisenbud1995commutative}.
One usually performs a final reduction step.
First remove every basis element whose leading monomial is divisible by the leading monomial of another basis element.
In particular, if two basis elements have the same leading monomial, one of them is redundant and must be removed after reduction.
Then reduce each remaining basis element by all the others and normalize leading coefficients to one.
The result is the reduced Gröbner basis.

The key computational property is then that every $f\in R$ has a unique normal form after division by this Gröbner basis~\cite{cox1997ideals,greuel2008singular}.
A monomial $z^\alpha$ is called standard if it is not divisible by any $\operatorname{LM}(g)$ with $g$ in the Gröbner basis.
Let $\mathcal B_I$ be the set of standard monomials.
The residue classes of the monomials in $\mathcal B_I$ form an $\F$-basis of $R/I$~\cite{cox1997ideals,greuel2008singular}.
Therefore, if $\mathcal B_I$ is finite, then
\begin{equation}
	\dim_{\F}R/I=
	\#\mathcal B_I.
	\label{eq:quotient-ring-standard-monomials}
\end{equation}

Here is a small example over $\mathbb F_2$ where the initial generators are not already a Gröbner basis.
Let
\begin{equation}
	R=\mathbb F_2[x,y],
	\qquad
	I=\langle x^2,\ y^2+x^2\rangle .
\end{equation}
Use lexicographic order with $x>y$.
Set
\begin{equation}
	f_1=x^2,\qquad f_2=y^2+x^2 .
\end{equation}
Both generators have leading monomial $x^2$:
\begin{equation}
	\operatorname{LM}(f_1)=x^2,\qquad \operatorname{LM}(f_2)=x^2 .
\end{equation}
Their $S$-polynomial is
\begin{equation}
	S(f_1,f_2)=f_1+f_2=y^2 .
\end{equation}
The monomial $y^2$ is not divisible by the current leading monomial $x^2$, so the remainder is nonzero.
Buchberger's algorithm therefore adjoins
\begin{equation}
	f_3=y^2 .
\end{equation}
Now the leading monomials include $x^2$ and $y^2$.
The remaining $S$-polynomials reduce to zero by division with the current list, so the loop has produced a Gröbner basis.
There is still a final simplification.
The two elements $f_1$ and $f_2$ have the same leading monomial $x^2$; moreover, reducing $f_2$ by $f_1$ gives
\begin{equation}
	f_2+f_1=y^2=f_3 .
\end{equation}
Thus $f_2$ is redundant.
After removing this redundant generator, the reduced Gröbner basis is
\begin{equation}
	G_{\mathrm{red}}=\{x^2,\ y^2\}.
\end{equation}
The standard monomials are therefore
\begin{equation}
	\mathcal B_I=\{1,\ x,\ y,\ xy\}.
\end{equation}
Thus
\begin{equation}
	\dim_{\mathbb F_2}R/I=4.
\end{equation}

The module version is similar.
Let $R^n$ have basis $e_1,\ldots,e_n$, and let $N\subseteq R^n$ be a submodule generated by vectors $g_1,\ldots,g_t$.
A module monomial has the form $z^\alpha e_i$.
To run a Gröbner-basis algorithm, one must first choose an order on these module monomials.
This includes a convention for comparing the basis directions $e_1,\ldots,e_n$, such as a term-over-position or position-over-term order.
For example, in $R^2$ one may choose an order with $e_1>e_2$, so the leading module monomial of $xe_1+e_2$ is $xe_1$.
Divisibility also remembers the basis direction:
\begin{equation}
	z^\alpha e_i \mid z^\beta e_j
	\quad\text{means}\quad
	i=j\ \text{and}\ z^\alpha\mid z^\beta .
\end{equation}
Thus $xe_1$ divides $x^3e_1$ but does not divide $x^3e_2$.
In the module division algorithm, a current leading term can be cancelled by a basis element $g\in G$ only when $\operatorname{LM}(g)$ divides that leading module monomial in this sense.
Having the same polynomial monomial is not enough if the terms lie in different basis components.
The module Buchberger algorithm again computes a Gröbner basis $G_N$ by reducing module $S$-polynomials with this ordered basis and this divisibility rule~\cite{greuel2008singular}.
The leading module monomials $\operatorname{LM}(g)$, $g\in G_N$, determine the standard module monomials: these are the $z^\alpha e_i$ not divisible by any leading module monomial~\cite{greuel2008singular}.

In matrix computations, a submodule is usually given by the columns of a polynomial matrix; computing a module Gröbner basis of this column submodule and counting its standard module monomials gives the vector-space dimension of the corresponding cokernel whenever that dimension is finite~\cite{greuel2008singular}.

%

\end{document}